\documentclass[twocolumn,aps,10pt,nofootinbib, longbibliography,superscriptaddress,notitlepage,prx,floatfix,raggedbottom]{revtex4-2}
\usepackage[utf8]{inputenc}
\usepackage{physics} 
\usepackage{amsfonts}
\usepackage{graphicx}
\usepackage[dvipsnames]{xcolor}
\usepackage{algorithm}
\usepackage{algpseudocode}
\usepackage{amsthm}
\usepackage{amsmath}
\usepackage{amssymb}
\usepackage{blkarray}
\usepackage{mathtools}
\usepackage{quantikz}
\usepackage{verbatim}
\usepackage[mathscr]{euscript}
\usepackage{mathbbol}
\usepackage{geometry}
\newtheorem{theorem}{Theorem}
\newtheorem{lemma}[theorem]{Lemma}
\theoremstyle{definition}
\newtheorem{definition}{Definition}[section]

\newcommand{\braid}{\text{\footnotesize $\mathrm{BRAID}_4$}}
\newcommand{\braidtwo}{\text{\footnotesize $\mathrm{BRAID}_2$}}
\newcommand{\Ptot}{P_{\mathrm{tot}}}
\newcommand{\wt}{\mathrm{wt}}
\newcommand{\Mod}[1]{\ \mathrm{mod}\ #1}
\newcommand{\CNOT}{\text{\footnotesize $\mathrm{CNOT}$}}
\newcommand{\CZ}{\text{\footnotesize $\mathrm{CZ}$}}

\DeclareMathOperator{\Ima}{Im}

\usepackage{hyperref}
\hypersetup{
    colorlinks=true,
    citecolor=red
}
\begin{abstract}
Majorana-based quantum computation in nanowires and neutral atoms has gained prominence as a promising platform to encode qubits and protect them against noise. In order to run computations reliably on such devices, a fully fault-tolerant scheme is needed for state preparation, gates, and measurements. However, current fault-tolerant schemes have either been limited to specific code families or have not been developed fully.  In this work, we develop a general framework for fault-tolerant computation with logical degrees encoded into Majorana hardware. We emphasize the division between even and odd Majorana codes and how it manifests when constructing fault tolerant gadgets for these families. We provide transversal constructions and supplement them with measurements to obtain several examples of fault tolerant Clifford gadgets. For the case of odd codes, we give a novel construction for gadgets using quantum reference frames, that allows to implement operations that are forbidden due to parity superselection. We also provide a fault-tolerant measurement scheme for Majorana codes inspired by Steane error correction, enabling state preparation, measurement of logical operations and error correction. We also point out a construction for odd Majorana codes with transversal T gates. Finally, we construct a high rate quantum LDPC Majorana code with logical qubits. Our work shows that all necessary elements of fault-tolerant quantum computation can be consistently implemented in fermionic hardware such as Majorana nanowires and fermionic neutral atoms.
\end{abstract}

\begin{document}

\title{Fault tolerant Operations in Majorana-based Quantum Codes: Gates, Measurements and High Rate Constructions}
\author{Maryam Mudassar$^*$}
\affiliation{Joint Center for Quantum Information and Computer Science, University of Maryland and NIST, College Park, MD, USA}

\def\thefootnote{$^*$}\footnotetext{\href{maryammu@umd.edu}{maryammu@umd.edu}}
\author{Alexander Schuckert}
\affiliation{Joint Center for Quantum Information and Computer Science, University of Maryland and NIST, College Park, MD, USA}
\affiliation{Joint Quantum Institute, University of Maryland and NIST, College Park, MD, USA}
\author{Daniel Gottesman}
\affiliation{Joint Center for Quantum Information and Computer Science, University of Maryland and NIST, College Park, MD, USA}

\maketitle
\section{Introduction}
Quantum computation is inherently susceptible to errors, and significant effort has gone into developing fault-tolerant schemes, all of which fundamentally rely on quantum error correction. However, a persistent challenge is that the physical error rates of many existing platforms remain above the fault-tolerance threshold. While ongoing advances in hardware — including improvements in coherence times and gate fidelities — aim to address this, an alternative and complementary approach is topological quantum computing, which seeks to build hardware with intrinsic protection against local errors.
Majorana zero modes (MZM) are believed to form at superconductor-semiconductor interface, and while progress in the discovery of MZM's has been challenging, there have been some promising efforts \cite{Aghaee2025, Liu2024} and many clear proposals on their implementation \cite{Alicea2012,aasen2025roadmapfaulttolerantquantum}. These modes are robust to local decoherence, hence one can encode logical information into these modes \cite{Kitaev_2001}. Majorana zero modes also correspond to Ising anyons \cite{Pachos_2012}, and one can realize logical gates on the encoded data by topologically protected physical operations. However, Ising anyons are not universal and in order to do universal computing, other topologically unprotected schemes are required \cite{Nayak_Measurememt}. Moreover, existing measurement based schemes still suffer from other noise processes \cite{alase2025decoherencemajoranaqubits1f} and require active error correction \cite{aasen2025roadmapfaulttolerantquantum}. 

The necessity to perform active error correction on these Majorana modes opens up the field of Majorana error correction, or Majorana codes. 
Majorana codes are stabilizer codes that encode logical qubits or logical fermionic degrees of freedom, with physical Majorana degrees of freedom. Since their introduction \cite{bravyi_majorana_2010}, there have been families of such codes, such as small codes \cite{hastings2017smallmajoranafermioncodes}, Majorana color codes \cite{Litinski_Majorana, McLauchlan2024newtwistmajorana, LiColorCode} and Majorana surface codes \cite{Tran_optimizing,PaetznickFloquetCodes,LiThreshold} being the important ones, where some codes have been optimized because proposed implementations allow local low weight measurements \cite{PaetznickFloquetCodes, Chao_2020}, or to protect against quasiparticle errors \cite{Kundu:2024auq}. The structure of the Majorana Clifford group has also been explored recently \cite{bettaque2024structuremajoranacliffordgroup, MudassarEncoding, McLauchlan_2022} which allows one to perform Clifford operations on such codes, enabling code initialization and logical gates.  

In the context of qubits based on Majorana nanowires, previous work \cite{aasen2025roadmapfaulttolerantquantum} has focused on encoding qubits into small superconducting islands composed of two or three nanowires, and then concatenating these small Majorana codes with known qubit codes such as surface codes or Floquet codes. While these designs are valuable since they can use existing techniques of fault-tolerant operations on qubit codes, these leave the benefits general Majorana codes afford on the table as well. Moreover, generally qubit-based codes are only tailored to Pauli noise models and would not be able to protect against native fermionic noise. In nanowire-based systems, quasiparticle poisoning can also cause gates to fail and measurements to become faulty, so the reliance on topological protection to suppress errors, especially during measurements becomes precarious. Hence, creating a general toolkit of fault-tolerant gadgets for these devices that does not rely on physical means of protection is crucial. 

Previous studies have developed fault-tolerant gate implementations for Majorana color codes that have utilized code deformations and lattice surgery techniques \cite{McLauchlan2024newtwistmajorana, LitinskiUniversal, LiColorCode} supplemented with magic state distillation \cite{BrienFermionMsd}. However, lattice surgery techniques require $d$ rounds of syndrome extraction, hence exploring other means of performing fault tolerant gates which includes transversal implementations and measurements is also extremely valuable. 

There have also been recent breakthroughs in manipulating fermionic neutral atom arrays in 2D \cite{yan_two-dimensional_2022,gonzalez-cuadra_fermionic_2023,bojovic2025}, and some proposals of encoding logical degrees of freedom into an error corrected subspace of these programmable arrays \cite{schuckert2024fermionqubitfaulttolerantquantumcomputing,ott2024errorcorrectedfermionicquantumprocessors}. In the absence of number conservation, these codes have been shown to be Majorana codes, further motivating the study of performing fault-tolerant operations on these devices. 

This work aims to fill some of the existing gaps in fault tolerance of Majorana codes. Section \ref{section: background} gives background on Majorana codes and known mappings to qubit codes, fermionic Clifford operations and implementation of fermionic operations in terms of measurements. Section \ref{section:errormodel} describes the error model for our codes, which is motivated from physical error sources in both neutral atom and nanowire architectures. We do not assume a specific code family; rather we divide them into two broad classes: even Majorana codes and odd Majorana codes in Section \ref{section: evenandodd}, and explore their structure and limitations. In Section \ref{section:reference} we explain the idea of reference systems that we use to carry out Clifford and non Clifford logical gates on odd Majorana codes throughout the paper. In Section \ref{section:transversalgates},  we discuss transversal Clifford operations; and supplement them with ancilla measurements to realize a fault-tolerant set of Clifford gates in Section \ref{section:nontransversal}. For fault tolerant error correction, we describe a Steane-inspired measurement gadget for Majorana codes in Section \ref{section:Steane}. In Section \ref{section:nonClifford} we describe a code that has transversal $T$ gate that can be useful in magic state distillation in a platform where one only has access to fermion-fermion Clifford operations.  Finally, we discuss Majorana CSS codes and Majorana qLDPC codes in Sections \ref{section: CSScodes} and \ref{section: qldpc code} respectively. We end with a discussion and possible future directions.

\section{Background}\label{section: background}
\subsection{Majorana fermions}
Majorana codes are designed for architectures with physical fermionic degrees of freedom. They are defined in terms of $2n$ Majorana operators:
\begin{align}
    \gamma_i =  a_i + a_i^{\dagger}\\
    \bar{\gamma}_i = -i(a_i - a_i^\dagger)
\end{align}
where $a_i,a_i^{\dagger}$ are the fermionic creation and annihilation operators fulfilling $\lbrace a_i,a_i^\dagger \rbrace=1$. Majoranas have anti-commutation relations
\begin{align}
\{\gamma_i,\gamma_j\} &= 2\delta_{i,j},\qquad \gamma_i^2 = I,\\
\{\gamma_i,\bar 
{\gamma}_j\} &= 0.
\end{align}
One can define a group $\text{Maj}(2n)$ for $2n$ Majorana modes, analogous to the Pauli group $P_n$, where the group is generated by Majorana monomials:
\begin{align}
    \text{Maj}(2n) 
    &= i^{r} \left\langle 
    \gamma_1^{c_0} \gamma_2^{c_1} \cdots \gamma_n^{c_{n-1}}\, 
    \bar{\gamma}_1^{c_n} \bar{\gamma}_2^{c_{n+1}} \cdots \bar{\gamma}_n^{c_{2n-1}}
    \right\rangle, \\
    &\text{with } r \in \left\{ 0,1,2,3 \right\},\ 
    c_i \in \left\{ 0,1 \right\}.
\end{align}
Similar to the Pauli group, the order of this group is $2^{2n+1}$. 
In MZM based architectures, each physical Majorana mode can be individually manipulated using gates, or Majorana bilinears can be measured. Hence logical states, commonly qubit states can be encoded into these modes and controlled. 

While Majorana zero modes and their manipulation was the original motivation for this paper, it is interesting to note that these codes also bear value for fermionic neutral atoms. 
In neutral atoms, fermionic sites are realized by sites of an optical lattice or a grid of optical tweezers~\cite{gonzalez-cuadra_fermionic_2023, schuckert2024fermionqubitfaulttolerantquantumcomputing}. Each of these sites can either be occupied by a fermionic atom (encoding the state $\ket{1}$) or not (encoding the state $\ket{0}$). Destroying or creating a fermion is then denoted with the complex fermionic operators $a_i$ and $a_i^\dagger$, respectively, and Majorana fermions are encoded into the real and imaginary parts of these complex fermionic operators. Similarly, electrons in dots and donors in silicon realize complex fermions and can be braided using tunneling~\cite{wang2022a} and pairing induced by the proximity effect with a superconductor. Moreover, qubit-fermion gates with hyperfine qubits in these platforms yield further flexibility~\cite{rad2024,schuckert2024fermionqubitfaulttolerantquantumcomputing}.

It is worth pointing out the normal physical setting for these systems corresponds to a particle number conservation constraint. More precisely, this means that physical Hamiltonian $H_{\text{f}}$ commutes with the total fermion number $N_{\text{tot}}$, which is a \textit{symmetry}, and hence $H_{\text{f}}$ can be written in a block diagonal form, and similarly, the Hilbert space can also be written in this block diagonal basis
\begin{align}
    \mathcal{H}_{f}=\bigoplus_{i=0}^{N} \mathcal{H}_i
\end{align}
Moreover, the native unitaries are also considered to be \textit{number conserving} operations. However, this is not an essential constraint, and people have shown proposals where the number conserving constraint can be lifted, either using reference systems \cite{ott2024errorcorrectedfermionicquantumprocessors} or by engineering specific \textbf{pairing} gates \cite{schuckert2024fermionqubitfaulttolerantquantumcomputing}. Doing so changes the symmetry  from $S=N_{\text{tot}}$ to $S=P_{\text{tot}}$, so $H_{\text{f}}$ is instead block diagonal in the total parity basis, and the Hilbert space $\mathcal{H}_f$ also becomes a $Z_2$ graded space. For our purpose, this means that we can use Majorana modes to represent physical degrees of freedom instead of fermions, and apply parity-preserving unitaries to our system.
\subsection{Error model for Majorana nanowires and neutral atoms}\label{section:errormodel}
Majorana codes are a natural fit for fermionic hardware as well as architectures based on MZM's. This is because they protect logical information against the native fermionic noise model. 
In \cite{Knapp2018modelingnoiseerror}, the authors describe a stochastic Majorana noise model. Analogous to Pauli noise models, a Majorana noise process can be modelled at different levels of complexity: from phenomenological noise to circuit level noise. In \cite{Knapp2018modelingnoiseerror}, the authors divide their Majorana architectures into $n/m$ islands where each island hosts $2m$ MZM's, and then consider noise processes that correspond to single mode errors, correlated errors and measurement errors. For general Majorana codes, modes are not always organized in terms of islands or qubits. We can describe a typical Majorana noise model using the following processes
\begin{itemize}
    \item \textit{Quasiparticle poisoning}: A single $\gamma_i$ error can act on mode $i$ and flip the state. This can occur with interaction from an external fermionic bath, or coupling between multiple modes or codeblocks. When Majoranas are organized in terms of islands \cite{Knapp2018modelingnoiseerror} or in plaquettes \cite{Litinski_Majorana}, these processes may be reduced using Coulomb blockades. Importantly, it has been suggested  in \cite{Knapp2018modelingnoiseerror} that in the absence of quasiparticle poisoning, the Majorana noise model reduces to a Pauli noise model. 
    \item \textit{Dephasing noise} Two Majoranas $\gamma_i,\bar{\gamma}_i$ corresponding to the same fermion mode may experience an error at the same time. In \cite{Knapp2018modelingnoiseerror}, this would correspond to a local error in the same island. A two-mode error in a single island or codeblock is mathematically equivalent to a single Pauli error, and can correspond to either $X$ or $Z$ depending on the basis choice.  
    \item \textit{Correlated/circuit level noise} If a $k$ mode gate occurs, it can spread a single error to at most $k$ modes, and this corresponds to a correlated error. In \textit{transversal} schemes such as ours, each $k-mode$ operation is broadcast over $k$ blocks, so correlated errors are curbed, while in other schemes such as \cite{aasen2025roadmapfaulttolerantquantum}, measurements are chosen so they do not spread errors.
\end{itemize}
Acting on the state, for a \textit{stochastic Majorana model}, time is discretized into finite steps $\tau$, and at each time step, the state undergoes:
\begin{align}
    \ket{\psi}\rightarrow \gamma_1^{a_1}\gamma_2^{a_2}\ldots \bar{\gamma}_{n}^{a_{2n}} \ket{\psi}
\end{align}
There may also be measurements, either mid-circuit or at the end, and these correspond to the measurement noise vector $\textbf{b}=(b_1,b_2\ldots b_n)$ added to the measurement vector $\textbf{m}=(m_1,m_2\ldots m_n)$. To correct for measurement noise, one can use classical decoding. We demonstrate this for fault-tolerant error correction in Section \ref{section:Steane}.

The error model of the neutral-atom-encoded Majorana fermions is the same as the one for MZMs~\cite{gonzalez-cuadra_fermionic_2023,schuckert2024fermionqubitfaulttolerantquantumcomputing, ott2024errorcorrectedfermionicquantumprocessors}: quasiparticle poisoning happens when an atom is lost, which realizes a noise operator $a_i$ which can be expanded in terms of $\gamma_i$ and $\bar \gamma_i$. Dephasing noise is induced by laser fluctuations of the optical lattice and optical tweezers and those errors may be correlated.

\subsection{Majorana codes}

A Majorana stabilizer code is defined by specifying the stabilizer group as the Abelian subgroup of $S_{\text{Maj}} \subseteq \text{Maj}(2n)$, s.t. $-I\notin S_{\text{Maj}}$. A Majorana stabilizer code on $2n$ modes corresponds to $n$ physical fermions, has $n-k$ stabilizer generators, and the code distance $d$ is the weight of the smallest logical operator. As a result, a Majorana code can also be written as $[[n,k,d]]_f$. 
Sometimes, an additional parameter called $l_\mathrm{even}$ is also defined since Majorana codes do not treat even and odd weight logical operators on the same footing, but in a fault tolerance context, this is not a good metric, see more discussion in Section \ref{section:transversalgates}. The logical operators for such codes may have even or odd weight, depending on whether $\Ptot\in S$ or not, see more discussion in Section \ref{section: evenandodd}.

Similar to the Pauli group, the Majorana group can be converted into a symplectic vector space with the symplectic form $\Lambda_f$ \cite{chien2022optimizingfermionicencodingshamiltonian}, 
\begin{align}
\Lambda_f = I + C 
\end{align}
where $I$ is the identity matrix and $C$ is a constant $2n\times2n$ matrix of ones. One can then find whether Majorana operators commute by finding their symplectic inner product defined as 
\begin{align}
    \textbf{x} \odot \textbf{z} = \textbf{x} \Lambda_f \textbf{z}^T \Mod{2}
\end{align}

where \textbf{x} and $\textbf{z}$ correspond to length $2n$ Majorana binary strings. 
In general, the commutation relation for two Majorana binary strings \textbf{x,z} are:
\begin{align}
 \text{If }|\textbf{x}||\textbf{z}|+\textbf{x}.\textbf{z}  \mod 2 =0, \text{they commute}\\
 \text{If }|\textbf{x}||\textbf{z}|+\textbf{x}.\textbf{z}  \mod 2 =1, \text{they anticommute}
\end{align}
In this paper, we make use of both odd and even Majorana strings, and the commutation rules for these are summarized:
\begin{itemize}
    \item If two Majorana strings have even weight, they commute if they overlap on an even number of sites, and anti-commute otherwise.
    \item If one Majorana string is even and the other is odd, they commute if they overlap on an even number of sites and anticommute otherwise. 
    \item If two Majorana strings have odd weight, they commute if they overlap on an odd number of sites, and anticommute otherwise. 
\end{itemize}

The \textit{parity superselection rule} is also an important constraint to take into account when considering Majorana codes. It states that one cannot create a state which admits a coherent superposition of an even and odd number of fermions. This means that all physical operations with respect to some \textit{frame of reference} must preserve the parity of the system. More formally, any quantum state $\rho$ can be decomposed as \cite{VidalSuperselection}
\begin{align}
    \rho = \rho_{e}\oplus \rho_{o}
\end{align}
such that $\Ptot\rho_{e}\Ptot^{\dagger}=\rho_{e}$, and $\Ptot\rho_{o}\Ptot^{\dagger}=\rho_{o}$. 
Other observables and unitaries must also must admit a direct sum structure with respect to the fermionic Hilbert space, which entails that they must commute with the total parity (see Appendix \ref{section: parityconstraintunitary} for proof). In particular, this constrains the stabilizer generators: they must commute with the total fermion parity i.e. $\Ptot = i^n \prod_{i=1}^n \gamma_i \bar{\gamma}_{i} $ since they correspond to physical, measurable operators. However, it is important to recognize that the superselection rule is defined with respect to a reference, and one can use \textit{reference systems}~\cite{BartlettSuperselection} to relax this constraint.

In general, logical operators may or may not commute with $\Ptot$. If $\Ptot\in S$, logical operators have to commute with $\Ptot$ and must therefore be even weight in the physical Majorana operators. Such codes are known as \emph{even} codes.  Codes that do not contain $\Ptot\in S$ have at least one logical operator which is odd weight~\cite{bravyi_majorana_2010}. Such codes are called \emph{odd} codes. As an example for an \emph{odd} code, consider the 1D Kitaev chain on $2n$ Majoranas. One can encode one logical fermion into the Majorana zero modes as follows:
\begin{align}
    \Gamma_1 = \gamma_1 \quad
    \Bar{\Gamma}_1 = \bar{\gamma}_1,
\end{align}
where we denoted the logical Majorana operators with $\Gamma_1$  and $ \Bar{\Gamma}_1$.
Here the code parameters are $[[n,1,1]]_f$. Note that while the distance is one, the Kitaev chain can still detect and correct all local errors that occur in the bulk of the chain.

One can also combine two Kitaev chains to obtain an \emph{even} code, also known as a Majorana tetron code~\cite{Litinski_Majorana}. Because even-weight operators commute between different codeblocks, this enables encoding a  logical \emph{qubit} by $\bar{X}=i\gamma_1\gamma_2$ and $\bar{Z}=i\gamma_1\bar{\gamma_1}$, and introducing an additional parity constraint $S=\gamma_1\bar{\gamma_1}\gamma_2\bar{\gamma_2}$. We give a longer discussion of even and odd codes in Section \ref{section: evenandodd}.

Majorana codes can be mapped to qubit CSS codes \cite{bravyi_majorana_2010} and vice versa, and these mappings are often a useful tool when finding Majorana codes. Moreover, classical weakly self-dual codes can also be converted to Majorana codes \cite{vijay2017quantumerrorcorrectioncomplex}. We review these as follows.
\begin{lemma}
(Vijay et al.~\cite{vijay2017quantumerrorcorrectioncomplex}) Any classical weakly self-dual $[2n,k,d]$ code can be converted to a Majorana code $[[n,n-k,d^{\perp}]]$ where $d^{\perp}$ is the dual distance.
\label{lemma: lemma1}
\end{lemma}
\textit{Short proof:} For a classical code to be weakly self-dual, $C\subseteq C^{\perp}$, implying $GG^{T}=0$, where $G$ is the generator matrix of the classical code. For the Majorana code, the generator matrix $G$ can be converted to the binary matrix of the stabilizer for the Majorana code.  The number of physical fermions are half the number of classical bits, the rows of generator matrix become the rows of the stabilizer, hence the number of logical degrees of freedom are $n-k$. The distance of the code is the weight of smallest element in $N(S)/S$, which is the weight of smallest element in $H\setminus G$. If the code is non-degenerate, as in the case considered in \cite{vijay2017quantumerrorcorrectioncomplex}, then $d^{\perp}$ is the code distance, else $d\geq d^{\perp}$. 

\begin{lemma}
(Bravyi et al.~\cite{bravyi_majorana_2010}) Any Majorana code $[[n,k,d]]_f$ can be mapped to a qubit weakly self-dual  CSS code $[[2n,2k,d]]$.
\end{lemma}
\textit{Short proof:} 
Set $H_x=S$, $H_z=S$ for the quantum CSS code, where $S$ is the stabilizer of the Majorana code. $H_xH_z^T=0$ since $H H^T=0$. The new qubit code has support on $2n$ qubits, and has $2n-2k$ generators. Distance remains the same since any vectors in $C^{\perp}\setminus C$ correspond to logical operators in both the Majorana code and the qubit CSS code. 

Note that Lemma 1 and Lemma 2 can be combined to obtain a weakly self dual CSS code from a classical code, as outlined in \cite{Steane96}. Majorana codes can also be obtained from qubit CSS codes.

\begin{lemma}(Bravyi et al.~\cite{bravyi_majorana_2010})
Any qubit weakly self-dual CSS code $[[n,k,d]]$ can be converted to a Majorana code $[[n,k,d]]_f$.     \label{lem_csstomaj}
\end{lemma}
\textit{Short proof:} Weakly self-dual for a CSS code implies $H_x=H_z=H$. Set $H_\gamma= H$ and $H_{\bar{\gamma}} = H$. This is a weakly self-dual Majorana CSS code. All the parameters are identical since all we have done is just map $X$ to $\gamma$ and $Z$ to $\bar{\gamma}$.  

Throughout the paper, we will be using these Lemmas extensively. 

\subsection{Majorana Clifford gates}
Majorana string operators are defined as:
\begin{align}
    M= i^r \gamma_1^{c_1}\ldots \gamma_n^{c_n} \bar{\gamma}_1^{c_{1+n}}\ldots\bar{\gamma}_n^{c_{2n}}, r=\{0,1,2,3\}, c_i=\{0,1\}
\end{align}
We can therefore represent a Majorana operator up to phases as a binary string $\textbf{c}\equiv (c_1,c_2\ldots c_{2n})$. In analogy to qubit Clifford gates, fermionic Clifford gates are defined as gates which map one Majorana string to a single other Majorana string.

Our two elementary Clifford gates are $\braidtwo$ 
\begin{align}
    \braidtwo(\bar{c}) = \exp(-\frac{\pi}{4} \bar{c}), \qquad \wt(\bar{c})=2,
\end{align}and $\braid$ 
\begin{align}
\braid(\bar{v})=\exp(i\frac{\pi}{4}\bar{v}), \qquad \wt(\bar{v})=4.
\end{align}
These gates act as
\begin{multline}
\braidtwo(\bar{c})(\bar{m})\braidtwo(\bar{c})^{\dagger} =\frac{1}{2}(1+(-1)^{\textbf{c}\Lambda_f \textbf{m}^T})\bar{m}+\\
\frac{1}{2}(-1+(-1)^{\textbf{c}\Lambda_f \textbf{m}^T})\Bar{c}\Bar{m}
\end{multline}
where $\Lambda_f=I+C$
\begin{align}
\braid(\Bar{v})[\bar{m}]\braid(\Bar{v})^\dagger) =\frac{i}{2}(1-(-1)^{\textbf{v}\Lambda_f \textbf{m}^T}) \Bar{v}\Bar{m}+\notag\\ \frac{1}{2}(1+(-1)^{\textbf{v}\Lambda_f \textbf{m}^T})\Bar{m}.
\label{eq:braidaction}
\end{align}
In a binary matrix form, $\braidtwo(\bar{c})$ can be written as
\begin{align}
\braidtwo(\Bar{c})_{\alpha,\beta}=
\begin{cases}
1 & \alpha = \beta ,\alpha,\beta \notin \Bar{c}\\
1 & \alpha \neq \beta,\alpha,\beta \in \Bar{c}\\
0 & \text{otherwise}
\end{cases},
\end{align}
and $\braid(\Bar{v})$ as
\begin{align}
\braid{(\Bar{v})}_{\alpha,\beta} = 
\begin{cases}
1 & \alpha=\beta, \alpha,\beta \notin \Bar{v} \\
1 & \alpha,\beta \in \Bar{v}, \alpha \neq \beta \\
0 & \text{otherwise}.
\end{cases}
\end{align}

The $\braid$ gate can also be understood as the generator of the logical braid group. The braid group (also known as the Artin braid group \cite{Kauffman_2018}), can be understood in terms of the $n-1$ braid generators $\sigma_i$ satisfying the relations
\begin{align}
    \sigma_i \sigma_j = \sigma_j \sigma_i \text{ for } |i-j|\geq2\\
    \sigma_i \sigma_{i+1}\sigma_i = \sigma_{i+1}\sigma_{i}\sigma_{i+1}.
\end{align}
For the $\braidtwo$ gates, the above are known to be satisfied. To see that these relations are also satisfied by the $\braid$ group, one can  define a group of logical operators:
\begin{align}
    \langle \Gamma_1, \Gamma_2, \ldots \Gamma_k\rangle
\end{align}
where $k=n/2$ such that 
\begin{align}
    \braid (\Gamma_1)\braid^{\dagger} = \Gamma_2 \quad \braid (\Gamma_2)\braid^{\dagger} = -\Gamma_1
\end{align}
It can also be shown that they satisfy the following relations
\begin{align}
    \sigma_i \sigma_j = \sigma_j \sigma_i \text{ for } |i-j|\geq2\\
    \sigma_i \sigma_{i+1}\sigma_i = \sigma_{i+1}\sigma_{i}\sigma_{i+1}
\end{align}
As an example, we can choose $\sigma_1 = \braid(1,2,3,4)$, $\sigma_2=\braid(2,3,4,5)$ and $\sigma_3=\braid(5,6,7,8)$, so it is obvious that two generators with indices two or more away do not overlap and hence commute. The proof for the second property is given in Appendix \ref{AppendixB}. A general $\braid$ generator can thus be constructed by conjugating a generator from the above list with $\braidtwo$ generators. 

\subsection{Implementation of the $\braidtwo$ and $\braid$ gates}
Throughout this paper, we will extensively use both $\braidtwo$ and $\braid$ gates as the physical Clifford operations. In this section, we discuss how these gates may be implemented in specific hardware.

\subsubsection*{Majorana nanowires}
In nanowires, measurement-based operations instead of unitary operations have been proposed to be the simplest to implement~\cite{Tran_optimizing, aasen2025roadmapfaulttolerantquantum}. To this end, one may write both $\braidtwo$ and $\braid$ gates in terms of projective measurements of 2-MZM or 4-MZM parities. A $\braidtwo$ gate can be expressed in terms of projectors, for example \cite{Tran_optimizing}:
\begin{align}
    \Pi_{+1}^{(34)}\Pi_{+1}^{(13)}\Pi_{+1}^{(23)} \Pi_{+1}^{(34)} = \Pi_{+1}^{(34)} \otimes \braidtwo(1,2)
\end{align}
here, $\Pi_{s}^{ij}=\frac{1}{2}(\mathbb{1}+is\gamma_i\gamma_j)$.

Note that the procedure above uses two ancilla Majorana modes but disentangles them at the end. Similarly, to implement a $\braid$ gate, one can prepare an ancilla state in the zero state
\begin{align}
    i \gamma_4 \gamma_5 \ket{\bar{0}} = \ket{\bar{0}}
\end{align}
so that one may implement a $\braid$ gate using the following projective and corrective operations
\begin{align}
    2 \,\braidtwo(2,5) \Pi_{-1}^{(24)}\Pi_{+1}^{(0134)} = \Pi_{+1}^{(45)}\otimes \braid(0,1,2,3)
\end{align}
where $\Pi_{+1}^{(ijkl)}=\frac{1}{2}(1+\gamma_i\gamma_j\gamma_k\gamma_l)$. Note that both operations reset the ancilla pair back to the original state. Therefore, unitary approaches that employ physical $\braidtwo$ and $\braid$ gates are theoretically equivalent to approaches involving 2-MZM and 4-MZM mid-circuit measurements and one can be compiled exactly in terms of the other. 

In \cite{aasen2025roadmapfaulttolerantquantum}, a measurement-based approach is proposed to perform single-qubit and two-qubit logical Clifford operations, but they depend on the ability to perform 2-MZM and 4-MZM measurements in a robust manner, along with using ancilla. To perform long range gates, one often needs an ancilla resource state in a logical GHZ which allows joint measurements of observables. Fault tolerant measurements then require a lattice surgery based approach, such as the one in highlighted in \cite{Litinski_Majorana}. The focus of our investigation is to expand the gate set by using Clifford operations available to the codes and building transversal gadgets for these codes.

\subsubsection*{Neutral atoms}

In neutral atoms, all braiding operations can be implemented unitarily.

The $\braidtwo$ gate in neutral atoms is implemented in two different ways, depending on whether the Majorana fermions on the same site or two different sites are braided. The former is realized by a local phase gate
$\exp(-i\frac{\pi}{2}n_i)$, where $n_i=a_i^{\dagger}a_i=\frac{1}{2}(1-i\bar{\gamma}_i\gamma_i)$. The latter is realized by decomposing the braiding gate into a tunneling and pairing gate~\cite{schuckert2024fermionqubitfaulttolerantquantumcomputing} 
\begin{align}
 &\exp(i\frac{\pi}{4} \bar \gamma_i \gamma_j)\notag\\&=\exp\left(i\frac{\pi}{4} \left(a_i^\dagger a_j +h.c.\right)\right)  \exp\left(i\frac{\pi}{4} \left(a_i^\dagger a_j^\dagger +h.c.\right)\right).   
\end{align} The tunneling gate is most easily implemented by either lowering the barrier between two sites in an optical lattice or by moving two tweezers close to each other, inducing tunneling~\cite{yan_two-dimensional_2022,gonzalez-cuadra_fermionic_2023}. The pairing gate is implemented by using bosonic molecules in a Bose-Einstein condensate, which provide a reservoir for creating and destroying fermion pairs~\cite{schuckert2024fermionqubitfaulttolerantquantumcomputing}.

The $\braid$ gate in neutral atoms is implemented using the density-density interaction $\exp(-i \pi n_in_j)$ and phase gates by decomposing
\begin{align}
    \exp(i\frac{\pi}{4}\bar{\gamma_i}\gamma_i \bar{\gamma}_k\gamma_k ) = \exp(-i \pi \left(n_i-\frac{1}{2}\right)\left(n_j-\frac{1}{2}\right)).
    \label{eq:braidgate}
\end{align}
Density-density interactions are implemented by using Rydberg interactions~\cite{jaksch2000,gonzalez-cuadra_fermionic_2023}. For a general $\braid$ gate, we can conjugate \ref{eq:braidgate} with $\braidtwo$ gates to get the desired form.

\section{Even and odd Majorana codes}\label{section: evenandodd}
 In qubit quantum codes, logical qubits are encoded into the +1 eigenspace of the stabilizer operators. Each code block can encode multiple logical qubits, but operators in one code block usually only have support over the corresponding block and adding multiple blocks corresponds to a tensor product of each block. In Majorana codes however, logical operators can be even or odd weight, and once multiple code blocks are considered, they no longer admit a tensor product structure. 
The 1D Kitaev chain previously thought to encode qubits, fails to do so when multiple chains are considered. We explore the consequences of even and odd codes in the following section, and further give a scheme using quantum reference frames to perform logical gates on odd codes. 
\begin{definition}
    An \textbf{even} code is defined as a code where $\Ptot\in S.$
\end{definition}
Let us look at a few properties of these even codes that come as a consequence:
\begin{itemize}
    \item All logical operators must have \textbf{even} weight.
    \item \textbf{Single logical per block:} If logical operators are constrained to a single block, one can only encode logical qubits here. If logicals are allowed to be non-local across blocks, then one can encode logical fermions.
    
    \item \textbf{Multiple logicals, single block:} One can construct either a set of logical qubits such that
    \begin{align}
        \{\bar{X_i},\bar{Z_i}\} = 0 \qquad [\bar{X_i},\bar{X_j}]=0
    \end{align}
    or a set of even weight logical fermions
    \begin{align}
    \{\Gamma_i,\bar{\Gamma}_j\} = 2\delta_{ij}
    \end{align}
    Or both, such that some degrees correspond to logical qubits, while others correspond to logical fermions. We can do this since this is just a basis choice, so logical fermions are related to logical qubits via a logical Jordan-Wigner transformation. 
   Note that we cannot inter-convert between fermionic and qubit terms using any set of physical operations. Unitary operations preserve the commutation relations, and it can be shown that such maps, while trace-preserving, are not completely positive, so no physical process can convert between the two. See Appendix \ref{section:JordanWignermap} for more details. 
   \item \textbf{Multiple logicals, multiple blocks:}
   Logical fermions with even weight in a block commute with logical fermions in a different block. In order to get actual logical fermions between multiple blocks, another non-local basis must be chosen that corresponds to a mutually anti-commuting set of operators.
\end{itemize}
The most familiar setting for fault tolerance, which is the one we use here, is having multiple blocks where each logical is confined to a single block. As a result, it is most natural to consider \textbf{logical qubits} to be encoded per block in even codes.

\begin{definition}
    An \textbf{odd} code is defined as a code where $\Ptot\notin S.$
\end{definition}
This entails the following properties for odd codes:
\begin{itemize}
\item Logical operators can be both \textbf{odd} and \textbf{even} weight. There must also be at least two odd weight logical operators. 

\item \textbf{Single logical degree per block:} If each logical is constrained to a single block, then one can choose two odd weight logical operators corresponding to a logical fermion. If logicals are allowed to be non-local, then one can encode logical qubits. 

\item Odd weight logicals can also correspond to logical qubits and can be chosen so they satisfy Pauli algebra. However, between blocks, they do not overlap and hence will fail to have the correct commutation relations. We can choose new logicals that fulfil the correct Pauli algebra, but these will be non-local between blocks now. We will explain the issues in greater detail using the Kitaev chain later. 

\item \textbf{Multiple logicals, single block:} Since a combination of odd and even logical operators within the same block is possible, one can choose all the odd operators to form a basis for logical fermions, and all the even weight operators to form a basis for logical Paulis. However, odd operators cannot be converted to even operators and vice versa in a closed system because of \textit{parity superselection rule}.

\item \textbf{Multiple logicals, multiple blocks:} If each logical operator is chosen to be local to a block, one can ask, what is the number of logical fermions per block? As specified, this depends on the chosen logical basis. The minimum number of logical operators with odd weight is two, if $\Ptot$ is chosen as one of the logical operators, and the rest except two operators corresponding to $\Gamma_1$ and $\bar{\Gamma}_1$ are even.
This corresponds to one logical fermion, and $k-1$ logical qubits.
To increase the number of logical fermions, one can pick a different logical basis choice in the following way: pick $\Gamma_1$ and $\bar{\Gamma}_1$ as the two odd logical operators to begin with, then choose $2l$ operators from the qubits operators to convert them into fermions via a logical Jordan Wigner transformation, let's call these operators $q_1,q_2,\ldots q_{2l}$. Next, multiply $\bar{\Gamma}_1$ by operators $q_1q_2\ldots q_{2l}$, and multiply each qubit operator by $\bar{\Gamma}_1$. This gives a new list of operators that are all odd weight and are all anticommuting. As an example, consider the following code
\begin{align*}
    S=\gamma_3\gamma_4\gamma_5\gamma_6\gamma_7\gamma_8\\
    \Gamma_1 = \gamma_1 \quad \bar{\Gamma}_1 = \gamma_2\\
    \bar{X}_2 =i \gamma_3 \gamma_4 \quad \bar{Z}_2 = i\gamma_4 \gamma_5\\
    \bar{X}_3 = i\gamma_6 \gamma_7 \quad \bar{Z}_3 =i \gamma_7 \gamma_8
\end{align*}
With the above procedure, we can convert qubits into fermions
\begin{align*}
S=\gamma_3\gamma_4\gamma_5\gamma_6\gamma_7\gamma_8\\
\Gamma_1 = \gamma_1 \quad \bar{\Gamma}_1 = \gamma_2 \gamma_3 \gamma_5\\
\Gamma_2 = \gamma_2\gamma_3 \gamma_4 \quad \bar{\Gamma}_2 =  \gamma_2 \gamma_4 \gamma_5\\    
\bar{X}_3 = i\gamma_6 \gamma_7 \quad \bar{Z}_3 = i\gamma_7 \gamma_8
\end{align*}
\end{itemize}
An important departure now, from the even codes case, is that for the setting of logicals remaining local to a block and having multiple blocks, one can have either \textbf{logical qubits} or \textbf{logical fermions} per block, or \textbf{both}, and the number can be varied as per the prescription given above at the time of basis choice, with the caveat that at a minimum, each block contains \textbf{one} logical fermion. 

\subsection{Example: even and odd codes}
A single Kitaev chain is an instance of an odd code. Traditionally, 1D Kitaev chain has been used as a model to encode topologically protected qubits, with qubits encoded in the Majorana zero energy edge modes.
\begin{align}
    \bar{X_1}=\gamma_1 \qquad \bar{Z_1}=i\gamma_1\bar{\gamma}_1
\end{align}
where $\gamma_i,\bar{\gamma_i}$ denote the edge modes. In isolation, this encodes a single logical qubit. However, if we add one more chain, it is clear that the logical operators between multiple blocks anti-commute. In order to encode logical qubits, one needs a fermion-to-qubit encoding, e.g. Jordan-Wigner transformation. Under such a transformation, the new logical operators are given as
\begin{align}
    \bar{X_i}=\prod_{k<i}\gamma_{k}\bar{\gamma_k} \gamma_i \qquad \bar{Z_i}=i\gamma_{i}\bar{\gamma_i},
\end{align}
where we assume that the chains are attached to each other in a one-dimensional geometry. 
There are other encodings, e.g. in Bravyi Kitaev Superfast encoding, one finds:
\begin{align}
    \bar{X_i}=\prod_{k<i}\gamma_k\bar{\gamma}_k \gamma_{i}\gamma_{i+1} \qquad \bar{Z_i} =  \prod_{k\leq i}\gamma_k\bar{\gamma}_k 
\end{align}
In both cases, it is obvious that the value of each new block depends on the fermion parity of previous blocks. This is a common feature of \textit{ancilla free} fermion-qubit encodings, since the encoding needs to preserve the overall sign. This presents issues from a fault tolerance perspective, since each block can not be thought of as a parallel error-correcting code, which makes the framework of transversal gates non-trivial for these codes, since a fault in one codeblock can affect the value of all subsequent logical operators.  It also makes certain logical gates much harder to implement; for instance, performing a logical Hadamard on the $n^{\text{th}}$ qubit requires a $\mathcal{O}(n)$ depth physical circuit, which is not practical. As a result, we prefer to keep each logical operator local to a codeblock, and hence a Kitaev chain would encode a single logical fermion.

Two Kitaev chains can also be combined, in what is known as a Majorana tetron code. 
\begin{align}
    \bar{X} = i\gamma_1\gamma_2 \quad \bar{Z} = i\gamma_1 \bar{\gamma}_1
\end{align}
And the total parity of the two chains is fixed to the even parity sector
\begin{align}
    S=\gamma_1\bar{\gamma}_1\gamma_2\bar{\gamma}_2.
\end{align}
Here each block of four modes commutes with operators in a different block, so it can naturally only encode logical qubits. 
\section{Quantum reference frames for gates on odd codes}\label{section:reference}

We described above that Majorana systems are subject to parity superselection rules, which further entails that all physical operations must correspond to even parity operations. But one may ask, is there a way around these rules? In \cite{BartlettSuperselection}, the authors describe that superselection rules are mathematically equivalent to a lack of reference frame, and hence one can lift superselection rules by defining an appropriate reference system.

As an example, consider Alice and Charlie, where Alice has a state she considers to be a superposition of opposite parity values $\ket{\psi}_A =  \frac{1}{2}(\alpha\ket{0} + \beta \ket{1})$. Alice describes her states in the fermionic occupation number basis using her own phase reference, while Charlie's reference is rotated by a certain angle $\phi$. If Charlie knows this angle, then he can describe Alice's states by rotating them by the corresponding angle, i.e. $\ket{\psi}_C = e^{i\phi R}\ket{\psi}_A$, where $\phi$ is a $Z_2$ variable and $R$ is the rotation generator. However, if Charlie does not know his angle with respect to Alice, then he will perform an average over the angles $\phi$, and such an averaging will yield a completely mixed state. This leads to the loss of all coherence information, and is known as the parity superselection rule, where coherent superpositions of even and odd states are not allowed. On the other hand, Alice is allowed to prepare a state which is a coherent superposition of same $\Ptot$ states and Charlie will prepare the same state up to an irrelevant global phase. Such a state is termed a \textit{parity invariant state}. Note that the superselection rule arises because of an absence of a global \textit{reference system} relative to which both Alice and Charlie can measure their states. If both parties share a consistent reference frame, then they can each `sync' the changes the other party made without loss of coherence. For instance, one can prepare even and odd parity invariant states of the form:
\begin{align}
    \ket{\Psi_e} = \frac{1}{2}(\ket{0}_R\ket{0}_S+\ket{1}_R\ket{1}_S)\\
    \ket{\Psi_o} = \frac{1}{2}(\ket{0}_R \ket{1}_S + \ket{1}_R \ket{0}_S)
\end{align}
The subscript $R$ and $S$ indicate that states correspond to a product Hilbert space $H_R\wedge H_S$, where $\wedge$ is a wedge product between states $i$ and $j$
\begin{align}
    \ket{i}\wedge \ket{j}=-\ket{j}\wedge \ket{i}
\end{align}
We have written the states above in the eigenbasis of $P_R$ and $P_S$. One can also choose a different basis $\ket{i\oplus j}\ket{j}$ corresponding to the observables $\Ptot$ and $P_S$, which formally corresponds to an isomorphism from $H_R\wedge H_S$ to $H_\text{tot}\wedge H_S$. By making this basis choice, we are describing our system relative to the global degree, and in such a setting, the RF is not quantum, thus it is described as `dequantization' in the literature \cite{BartlettSuperselection}. Thus in this dequantized picture, we have:
\begin{align}
    \ket{\Psi_e}=\frac{1}{\sqrt{2}}(\ket{0}_{\Ptot}(\ket{0}_S+\ket{1}_S))\\
    \ket{\Psi_o}=\frac{1}{\sqrt{2}}(\ket{1}_{\Ptot}(\ket{0}_S+\ket{1}_S))
\end{align}
which means that \textit{relative} to the reference, one can now prepare coherent superpositions of even and odd states, and hence `violate' the parity SSR.

For multiple systems, such as Alice and Charlie in the above example, one can prepare 
\begin{align}
    \ket{\Psi_e} =\frac{1}{2}(\ket{0}_R(\ket{00}_{AC}+\ket{11}_{AC})+\ket{1}_R(\ket{01}_{AC}+\ket{10}_{AC})) 
    \label{eq:parityinvariantstate}
\end{align}
which after dequantization, yields:
\begin{align}
    \ket{\Psi_e}=\frac{1}{2}\ket{0}_{\Ptot}(\ket{0}+\ket{1})_A(\ket{0}+\ket{1})_C
\end{align}
thus proving that relative to this basis, coherent superpositions of multiple systems can be prepared. Moreover, additional states can be tensored in and be later used as ancillas while remaining uncorrelated to the reference. To prepare such a state, one can construct a circuit composed of $\braidtwo$ operations. More details in Appendix \ref{section:encodingcircuitforref}. 

We now describe performing gates on odd codes using the reference.
\begin{itemize}
    \item \textbf{Applying odd logical Majoranas:} Here, the example is of logical `$X$' gate. This corresponds to a bit flip on the system, thus we call this a logical `$X$'. Note that such an operation is not parity preserving, hence unphysical, but we can bring the reference here and define an operation that is parity preserving globally. 
    \begin{align}
    \bar{X}_{pp} = \gamma_R \gamma_A
    \end{align}
    For instance, if the initial state is:
    \begin{align}
        \ket{\psi}_{RAC}=\ket{000}
    \end{align}
    Applying logical `$X$' gives
    \begin{align}
        \bar{X}_{pp}\ket{\psi} =\ket{110}
    \end{align}
    Writing in the dequantized basis and tracing over the global state gives
    \begin{align}
        \Tr_{gl}[\bar{X}_{pp}\ket{\psi}\bra{\psi}\bar{X}_{pp}^{\dagger}] = \ket{10}\bra{10}_{AC}
    \end{align}
\item \textbf{Applying logical Hadamard}: Logical Hadamard is desribed strictly with respect to the reference:
\begin{align}
    \bar{H}_{pp} = \exp(-\frac{\pi}{4}\bar{\gamma}_A \gamma_R)
\end{align}
On the state, the action is given as
\begin{align}
    \bar{H}_{pp}\ket{000}_{RAC}=\frac{1}{2}(\ket{00} + \ket{11})_{RA}\ket{0}_C
\end{align}
Externalising the reference and tracing over, we get
\begin{align}
  \mathrm{Tr}_{\mathrm{tot}}[\bar{H}_{pp}\ket{000}_{RAC}\bra{000}\bar{H}^{\dagger}_{pp}] =  \ket{+}_A\bra{+}_A \ket{0}_C\bra{0}_C
\end{align}
On the operators, this gives $\bar{H}_{pp}\bar{X}_{pp}\bar{H}^{\dagger}_{pp}=\eta \bar{Z}_{pp}$ and $\bar{H}_{pp}\bar{Z}_{pp}\bar{H}^{\dagger}_{pp}=\xi \bar{X}_{pp}$ where $\eta,\xi =\pm 1$. 
\item \textbf{Applying two mode gates:} The reference scheme can be used to apply two mode operations that violate parity, but since the construction is more complex, we defer the discussion to Section \ref{section:transversalgates}.

\end{itemize}
\subsection{Limitations of reference frame construction}
This approach can be used to thus apply any odd weight logical operation. Note that the same reference mode is used for any odd logical operation, so an error in the reference mode will lead to the application of faulty gates. To make this mode less vulnerable, one can use code concatenation such that the reference mode is now a logical Majorana instead, and then error correction can be applied periodically to flush out the errors from the reference. Moreover, this approach does add temporal overhead almost inevitably, since every odd gate must utilize the same reference block. For instance, a parallel circuit of Hadamard operations turns into a serial set of operations acting between the reference block and the respective data block. In the worst case, it can thus turn $\bar{H}_{pp}^{\otimes n }$ into an $O(n)$ circuit. It is possible in some cases to reduce this overhead temporally but then may use ancilla modes such as in Figure: \ref{fig:coupledtetron} and hence correspond to a spatial overhead. 

We also note that a reference system construction was also developed in \cite{ott2024errorcorrectedfermionicquantumprocessors} to perform number non-conserving operations; however our construction requires a quantum reference system, hence the size does not scale with the number of modes in the data code block.

\section{Transversal gate gadgets} \label{section:transversalgates}

Typically for qubit codes, transversality is the first step towards fault tolerance because of its simple construction and since it prevents the propagation of errors which is vital for a fault-tolerant gadget. The notion of transversality is often variable: sometimes it corresponds to performing the same gate on all the physical qubits in the same block, which is sometimes referred to as \textit{strongly transversal} 
\cite{TexeiraStronglyTransversal}, whereas on other occasions the physical gate may vary on each qubit but still admit a tensor product decomposition, and it may also not act on each qubit too. In our framework, when considering gates acting on multiple blocks, we consider the former definition which corresponds to the same gate acting on each \textit{mode} in a code block. Another critical aspect when considering the code's protection is also the fact that it may treat even and odd errors separately, and while it may have good protection against even parity errors, it may be vulnerable to odd errors, and this is captured by $l_{\text{even}}$. However, there are issues when considering this parameter when one has multiple code blocks and transversal gates. Firstly, a physical process that results in an even parity error can affect a single code block, but it can also affect multiple code blocks, such that it acts as an odd parity error, so superselection does not actually help boost the code protection. (This was also alluded to in \cite{bravyi_majorana_2010}.) Secondly, if one considers a transversal gate between multiple blocks, it may fail and change an error which is even on one code block to an error that is odd on both. In other words, local parity conservation is too strict of a constraint for error models and gadgets to have, so the significance of $l_{\text{even}}$ for a block is minimal. 

The Clifford operations available to us are $\braidtwo$ and $\braid$. For a transversal operation, we require the stabilizer to be preserved, without increasing the weight of errors beyond what can be detected or corrected. We have $\braidtwo$ and $\braid$ gates at our disposal in Majorana codes, which acting on two or four copies of the code automatically preserve the code space of any Majorana code. 
\subsection{Detail of $\braidtwo$ and $\braid$ transversality}

Let us consider the action of $\braidtwo$ on two codeblocks. The stabilizer is given by $S=\langle M_i\otimes I, I\otimes M_i\rangle$, where $M_i$ are the generators on each individual block. A transversal $\braidtwo$, denoted as $(\braidtwo)^T$, acts on every mode $i$ in block 1 to every mode $i$ in block 2. For even weight operators, the transversal gate acts as a `swap' gate since the global sign disappears. This transforms the stabilizers as
\begin{align}
    (\braidtwo)^T (M_i\otimes I) (\braidtwo^{\dagger})^T = I\otimes M_i
\end{align}
On the logicals, this corresponds to
\begin{align}
    (\braidtwo)^T (\Bar{X}\otimes I) (\braidtwo^{\dagger})^T &= I\otimes \Bar{X}
    \\
    (\braidtwo)^T (\Bar{Z}\otimes I) (\braidtwo^{\dagger})^T &= I\otimes \Bar{Z}
\end{align}
Hence, the logical SWAP can be carried out transversally. 
\\
Let us next look at the $\braid$ gate. This is now acting on four blocks of the code, so the stabilizer is $S=\langle M_i \otimes I\otimes I\otimes I, I\otimes M_i\otimes I\otimes I, I\otimes I\otimes M_i\otimes I, I \otimes I \otimes I \otimes M_i \rangle$. On the logicals, a transversal $(\braid)^T$ corresponds to
\begin{align}
    (\braid)^T (\Bar{X}III) (\braid^{\dagger})^T = I\Bar{X}\Bar{X}\Bar{X}
\\
 (\braid)^T (\Bar{Z}III) (\braid^{\dagger})^T = I\Bar{Z}\Bar{Z}\Bar{Z}   
\end{align}

However, certain operators commute with this $(\braid)^T$ gate. For instance, note that 
\begin{align}
    (\braid)^T (\bar{X}\bar{X}II) (\braid^{\dagger})^T = \bar{X}\bar{X}II\\
    (\braid)^T (\bar{Z}\bar{Z}II) (\braid^{\dagger})^T = \bar{Z}\bar{Z}II
\end{align}
This implies that states prepared in a logical Bell state remain invariant under this transversal gate.
\subsection{Mode permutation to achieve logical gates}
Consider a Majorana code where $\bar{X}$ and $\bar{Z}$ are both even and have the same weight. One toy example is the Majorana tetron code, where weight of all three logical operators is the same. Hence, a single $\braidtwo$ gate suffices to perform a logical Hadamard or a logical phase gate, upto Paulis.
A single permutation gate, however, can also fail, and can potentially cause a weight one error to become a weight two error. To avoid this, we can apply an ancilla mediated swap \cite{gottesman2024surviving}, described as follows: prepare an additional ancilla codeblock, and now instead of swapping directly between the modes 1 and 2, one can have three swap gates for each swap gate. This way, even if one swap gate fails, weight-1 errors do not become weight-2 errors on the data block, c.f. Fig. \ref{fig:ancillamediatedswap}. Note that while a single braid gate is not strictly a swap gate, three braids preserve the braid structure i.e. we want $\gamma_1\rightarrow \gamma_2, \gamma_2\rightarrow -\gamma_1$, and it is easy to check that three braids below give us $\gamma_1\rightarrow \gamma_2, \gamma_2\rightarrow -\gamma_1, \gamma_a\rightarrow \gamma_a$. 
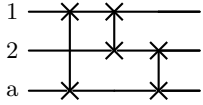
\begin{figure}
    \centering
\begin{quantikz}
\lstick{1} & \swap{2} & \swap{1}     &  & \qw \\
\lstick{2} &  & \swap{0}  & \swap{1} & \qw \\
\lstick{a} & \swap{0}     &        & \swap{0} & \qw
\end{quantikz}
\caption{Ancilla mediated swap. Note that typically $\braidtwo$ gates are considered robust to local perturbations since it is \textit{topologically protected}.}
\label{fig:ancillamediatedswap}
\end{figure}

A second example of permutation based gates comes from exploiting the automorphism structure of quantum codes. The automorphism group of a code $\text{Aut}(C)$ is defined by the set of logical gates that arise from permutation of modes such that they preserve the stabilizer group. For Majorana codes, one notable example of a code family with a non-trivial automorphism group is the Majorana Reed-Muller code, which are derived from applying Lemma \ref{lemma: lemma1} to classical Reed-Muller codes $\text{RM}(r,m)$ \cite{vijay2017quantumerrorcorrectioncomplex}. In order to exhibit weak self-duality, the parameters must satisfy the constraint that $m\geq 2r+1$. As is pointed out in \cite{vijay2017quantumerrorcorrectioncomplex}, the Majorana code obtained from an RM($1,m$) code $[n=2^m,m+1,d]$ is $[[2^m,2^m-m-1,d^{\perp}]]$.  

For example, $\text{RM}(1,3)$ saturates this bound, and in fact corresponds to the extended Hamming code. 
More generally, the $\text{Aut}(C)$ for Hamming codes is $\text{GL}(k,2)$, while the $\text{Aut}(C)$ of Reed-Muller codes is given by $\text{GL}(m,2)$. 

Consider the case of $\text{RM}(1,4)$. The parity check matrix is given:
\begin{align}
H=
\begin{bmatrix}
 1 & 1 & 1 & 1 & 1 & 1 & 1 & 1 & 1 & 1 & 1 & 1 & 1 & 1 & 1  & 1\\
 0 & 0 & 0 & 0 & 0 & 0 & 0 & 0 & 1 & 1 & 1 & 1 & 1 & 1 & 1  & 1\\
 0 & 0 & 0 & 0 & 1 & 1 & 1 & 1 & 0& 0& 0 & 0 & 1 & 1 & 1 & 1\\ 
 0 & 0 & 1 & 1 & 0 & 0 & 1 & 1 & 0 & 0 & 1 & 1 & 0 & 0 & 1 & 1\\
 0 & 1 & 0 & 1 & 0 & 1 & 0 & 1 & 0 & 1 & 0 & 1 & 0 & 1 & 0 & 1
\end{bmatrix}
\end{align}
The code parameters are $[[8,3,4]]_f$. The logical basis vectors can be found by:
\begin{align}
    H \Lambda_f \bar{L} = 0
\end{align}
i.e. vectors that commute with the stabilizer but are not contained inside the stabilizer.  This corresponds to:

\begin{align}
    \Bar{Z}_1 &= \gamma_{11} \gamma_{12} \gamma_{13} \gamma_{14}
    &\quad
    \Bar{X}_1 &= \gamma_5 \gamma_7 \gamma_{11} \gamma_{12} \gamma_{14} \gamma_{15} \\
    \Bar{Z}_2 &= \gamma_{13} \gamma_{14} \gamma_{15} \gamma_{16}
    &\quad
    \Bar{X}_2 &= \gamma_3 \gamma_5 \gamma_{11} \gamma_{14} \gamma_{15} \gamma_{16} \\
    \Bar{Z}_3 &= \gamma_2 \gamma_3 \gamma_5 \gamma_7 \gamma_{10} \gamma_{11} \gamma_{13} \gamma_{15}
    &\quad
    \Bar{X}_3 &= \gamma_5 \gamma_6 \gamma_9 \gamma_{11} \gamma_{14} \gamma_{15}
\end{align}
One can do the following permutation using $\braidtwo$ gates:
\begin{align}
    G=(5,7)(6,8)(13,15)(14,16)
\end{align}
This permutation preserves the stabilizer and transforms the logical operators as follows
\begin{align}
    \Bar{Z}_1 &\rightarrow \Bar{Z}_1 \Bar{Z}_2 
    &\qquad 
    \Bar{X}_1 &\rightarrow \Bar{X}_1 \Bar{Z}_2 \\
    \Bar{Z}_2 &\rightarrow \Bar{Z}_2
    &\qquad 
    \Bar{X}_2 &\rightarrow \Bar{X}_1 \Bar{X}_2 \Bar{Z}_1 \Bar{Z}_2 \\
    \Bar{Z}_3 &\rightarrow \Bar{Z}_3 
    &\qquad 
    \Bar{X}_3 &\rightarrow \Bar{X}_3
\end{align}
This corresponds to the logical gate 
\begin{align}
 G= \CZ(1,2)\CNOT(2, 1)  
\end{align}
Note that since every logical operator has even weight, we do not need to worry about the phase resulting from $\braidtwo$ gates. It is also interesting to note that swap operations in this code permit both $\CNOT$ and $\CZ$ gates, as opposed to permutation gates in qubit codes which correspond only to $\CNOT$ gates in the absence of further symmetries in the code structure.

One can also have Majorana codes with transversal non Clifford gates, but we defer the discussion of this code to Section \ref{section:nonClifford}. 

\section{Fault tolerant gate gadgets: supplementing transversal operations with measurements}\label{section:nontransversal}
Transversal gates are limited by the code structure and in order to do a full set of Clifford gates, we often need to supplement with other techniques. For instance, if the logical operators have unequal weights, such as in a Majorana color code where the $X$ and $Z$ logical operators are different weights, then permutations are no longer enough.  We suggest here to use more ancilla code blocks initialized in a particular state, along with measurements, which takes us out of the setting of just transversal gates.
\subsection{Even codes}
As a toy example, we can take two tetron codes and choose logicals that have support over both tetron blocks. The standard choice of logical operators for two tetrons are:
\begin{align}
    \bar{X}_1 = \gamma_{\mathrm{a}}\gamma_{\mathrm{b}}\quad \bar{Z}_1 = \gamma_{\mathrm{b}}\gamma_{\mathrm{c}}\\
    \bar{X}_2 = \gamma_{\mathrm{e}}\gamma_{\mathrm{f}} \qquad \bar{Z}_2 = \gamma_{\mathrm{f}} \gamma_{\mathrm{g}}
\end{align}
We instead pick logical operators to get a `coupled tetron' code
\begin{align}
    \bar{X}_1=\gamma_\mathrm{a}\gamma_\mathrm{b}\gamma_\mathrm{e}\gamma_\mathrm{f}\quad
    \bar{Z}_1=\gamma_\mathrm{a}\gamma_\mathrm{c}\\
    \bar{X}_2=\gamma_\mathrm{e}\gamma_\mathrm{f}\quad
    \bar{Z}_2=\gamma_\mathrm{a}\gamma_\mathrm{c}\gamma_\mathrm{e}\gamma_\mathrm{g}
\end{align}
which now gives us two logical operators but the $X$ and $Z$ weights are unequal. Below, we give a transversal scheme using ancilla to perform a logical Hadamard on the first qubit. Note that the code has distance two, so it can only detect weight one errors, and the objective of the scheme is to make sure that each codeblock still has weight one errors while preserving the stabilizer and enacting the desired logical gate. For the sake of this example we chose this code, but it can also be generalized for a higher distance code with perhaps different overheads on ancilla count and circuit depth.

A logical Hadamard is then implemented using the following procedure:
\begin{itemize}
    \item Take four more ancilla blocks, all encoded in the +1 eigenspace of the standard tetron code. Additionally, ancilla 2 and ancilla 3 are prepared in a logical Bell state, meaning that $\bar{X}_\mathrm{a_2}\bar{X}_\mathrm{a_3}$ and $\bar{Z}_\mathrm{a_2}\bar{Z}_\mathrm{a_3}$ are also in the stabilizer. The logical state of the first and last ancilla block can be arbitrary.  
    \item Now do a transversal $\braid$ gate between the first three ancilla and the first data block. This transforms $\Bar{X}_1$ to $\Bar{X}_2\Bar{X}_\mathrm{a_1}\Bar{X}_\mathrm{a_2}\Bar{X}_\mathrm{a_3}$ and $\bar{Z}_1$ to $\bar{Y}_{a_1}\bar{Y}_{\mathrm{a}_2}\bar{Y}_{a_3}$. By multiplying with the stabilizer, we get $\bar{X}_1\rightarrow\bar{X}_2\bar{X}_\mathrm{a_1}$ and $\bar{Z}_1\rightarrow \bar{Y}_{\mathrm{a}_1}$.   
    \item Do a second transversal $\braid$ gate between the first two blocks and the first and last ancilla. This transforms $\Bar{Z}_1$ to $\Bar{Z}_2\Bar{Z}_\mathrm{a_2}\Bar{Z}_\mathrm{a_3}\Bar{Z_\mathrm{a_4}}$, and then multiplying with stabilizer gives $\Bar{Z}_2\Bar{Y}_\mathrm{a_4}$, but leaves $\bar{X}_2\bar{X}_\mathrm{a_1}$ invariant. 
    \item The weight of both logicals is the desired weight, but we still want to do operations transversally that preserve the weight and the stabilizer. The third transversal $\braid$ converts $\bar{X}_1$ to $\bar{Z}_1\bar{X}_{\mathrm{a_1}}$ and $\bar{Z}_1$ to $\bar{X}_1\bar{Y}_{\mathrm{a_4}}$. The ancilla qubits can be measured out in the correct basis. The whole circuit is shown in Fig.~\ref{fig:coupledtetron}. 
    \item Note that this circuit only applies to this example, but similar techniques may be used to construct Hadamard for other codes with this property.
\end{itemize}
\begin{figure}[h]
    \centering
    \includegraphics[width=0.9\linewidth]{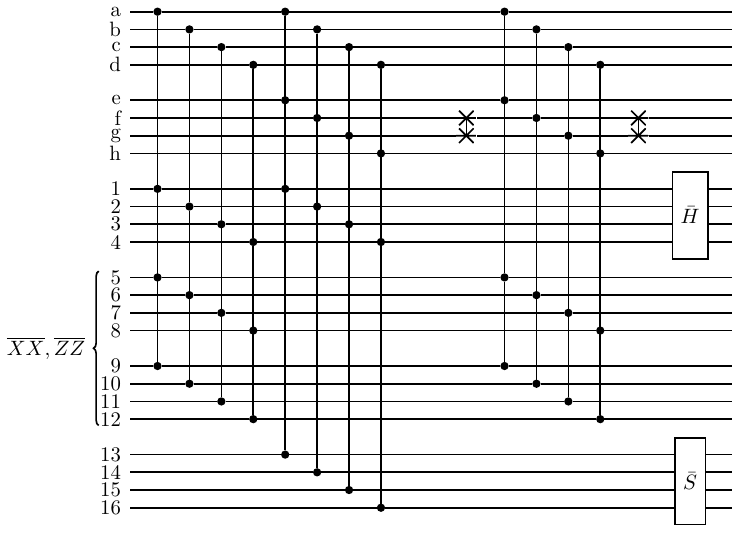}
    \caption{Logical Hadamard for coupled tetron, : $\bar{X}_1=\gamma_\mathrm{a}\gamma_\mathrm{b}\gamma_\mathrm{e}\gamma_\mathrm{f}$, $\bar{X}_2=\gamma_\mathrm{e}\gamma_\mathrm{f}$, $\bar{Z}_1=\gamma_\mathrm{a}\gamma_\mathrm{c}$, $\bar{Z}_2=\gamma_\mathrm{a}\gamma_\mathrm{c}\gamma_\mathrm{e}\gamma_\mathrm{g}$. The four mode gate denotes $\braid$ and two mode gate denotes the $\braidtwo$. The ancillas are then measured out in the occupation number basis i.e. eigenbasis of $\bar{Z}=i\gamma_2\gamma_3$ for each ancilla block.}
    \label{fig:coupledtetron}
\end{figure}
Another example is constructing transversal $\CNOT$ on even codes, akin to gadgets constructed in \cite{GottesmanTheoryofFaultTolerance}. Consider a circuit on four blocks, where the first two blocks are data blocks, labelled a and b and the last two blocks are ancilla blocks, labelled 1 and 2.
The circuit is given by the following unitary
\begin{align}
    U= &\quad\braidtwo^T(b,1)\braidtwo^T(a,b)\notag\\&\times\braidtwo
    ^T(a,2) \braid^T(a,b,1,2)
\end{align}
The circuit has the following logical action:
\begin{align*}
\bar{X}_\mathrm{a}\rightarrow \bar{X}_\mathrm{a}\bar{X}_\mathrm{b}\bar{X}_1\qquad \bar{Z}_\mathrm{a}\rightarrow \bar{Z}_\mathrm{a}\bar{Z}_\mathrm{b}\bar{Z}_1 \\
\bar{X}_\mathrm{b}\rightarrow \bar{X}_\mathrm{b}\bar{X}_1\bar{X}_2 \qquad
\bar{Z}_\mathrm{b}\rightarrow \bar{Z}_\mathrm{b}\bar{Z}_1\bar{Z}_2 \\
\bar{X}_1\rightarrow \bar{X}_\mathrm{a}\bar{X}_1\bar{X}_2 \qquad
\bar{Z}_1\rightarrow \bar{Z}_\mathrm{a}\bar{Z}_1\bar{Z}_2\\
\bar{X}_2\rightarrow \bar{X}_\mathrm{a}\bar{X}_\mathrm{b}\bar{X}_2 \qquad
\bar{Z}_2\rightarrow \bar{Z}_\mathrm{a}\bar{Z}_\mathrm{b}\bar{Z}_2 
\end{align*}
while the stabilizer is preserved. Initially, the ancilla blocks are prepared in the $\bar{Z_1}$ and $\bar{Z_2}$ state. Since we fix the state of the ancilla blocks, the ancilla blocks now correspond to stabilizer states, and $\bar{Z}_1$ and $\bar{Z}_2$ correspond to stabilizer elements.   At the end, we measure the ancilla in $\bar{X}_1$ and $\bar{X}_2$ basis. Since the $\bar{Z}$ logical representatives anticommute with the measurement, one can multiply with stabilizer elements of the ancilla states (for $\bar{Z_a}$ multiply by $\bar{Z}_1.\bar{Z}_2'$ and for $\bar{Z_b}$ multiply by $\bar{Z}_1$ where the primed operators represent transformed operators). This gives the desired transformation for a transversal CNOT. 
\subsection{Odd codes}

Another approach can be to encode logical fermions into the odd operators of the code, but it may be desirable to construct a CNOT between states in the number basis. For instance, while the physical terms in a closed fermionic Hamiltonian are in a fixed parity sector, it may be desirable in a simulation algorithm to perform a parity switching operation if the system is not closed and interacts with a fermionic bath. The logical CNOT in this context would mean:
\[
\begin{aligned}
    \Gamma_a &\rightarrow \Gamma_a \Gamma_b\\
    \Gamma_b &\rightarrow \Gamma_b\\
    \Gamma_a\bar{\Gamma}_a&\rightarrow \Gamma_a \bar{\Gamma}_a\\
    \Gamma_b\bar{\Gamma}_b& \rightarrow \Gamma_a \bar{\Gamma}_a\Gamma_b \bar{\Gamma}_b
    \label{eq:CNOT}
\end{aligned}
\]
Such an operation clearly switches parity between odd and even logical operators such as $\Gamma_a\rightarrow \Gamma_a \Gamma_b$, and thus would not be allowed without a reference. The odd operators are thus defined, using ideas from section \ref{section: evenandodd}
\begin{align}
    \Gamma_a' = \Gamma_a \Gamma_r\\
    \Gamma_b' = \Gamma_b \Gamma_r
\end{align}

Note that $\Gamma_a$ and $\Gamma_b$ anti-commute thus correspond to logical fermions, and similarly, $\Gamma_a'$ and $\Gamma_b'$ are also chosen to anti-commute.
In this case, we also need to fix the state of the ancilla, and here we choose it to be an eigenstate of
\begin{align}
    S_1 = i\Gamma_1\Gamma_r, \quad S_2 = i\Gamma_2 \bar{\Gamma}_r
\end{align}
The circuit we apply to get the desired ancilla state is
\begin{align}
    C = \braidtwo(\bar{2},\bar{r})\braidtwo(\bar{1},r)
\end{align}
where $\braidtwo(\bar{2},\bar{r})$ is a product of $\braidtwo$ gates between the barred Majorana operators in block two and the reference block (block $r$), and $\braidtwo(\bar{1},\bar{r})$ is a product of $\braid$ gates between the barred Majorana operators between block 1 and unbarred operators in block $r$. 

The initial stabilizer elements for the ancilla are $i\Gamma_1\bar{\Gamma}_1$ and $i\Gamma_2\bar{\Gamma}_2$, which can be conjugated by the above circuit $C$ to obtain the required state.

Next, a circuit consisting of transversal $\braid$ and $\braidtwo$ can now be executed as
\begin{align}
    U = \braidtwo^T(a,b)\braidtwo^T(b,2)\braid^T(a,b,1,2)
\end{align}
giving
\begin{align}
    \Gamma_a':\Gamma_a \Gamma_r \rightarrow \Gamma_a\Gamma_1 \Gamma_2 \Gamma_r,\\
    \Gamma_b':\Gamma_b \Gamma_r \rightarrow \Gamma_a\Gamma_b \Gamma_1 \Gamma_r,\\
    \Gamma_a\bar{\Gamma_a}\rightarrow \Gamma_a\bar{\Gamma_a} \Gamma_1\bar{\Gamma_1} \Gamma_2\bar{\Gamma_2}\\
    \Gamma_b\bar{\Gamma_b}\rightarrow \Gamma_a\bar{\Gamma_a} \Gamma_b\bar{\Gamma_b} \Gamma_2\bar{\Gamma_2}
    \end{align}
on the data block, while the ancilla operators transform as
\begin{align}
    \Gamma_1 \Gamma_r \rightarrow \Gamma_a\Gamma_b \Gamma_2 \Gamma_r,\\
    \Gamma_2 \bar{\Gamma}_r \rightarrow \Gamma_b\Gamma_1 \Gamma_2 \bar{\Gamma}_r.
\end{align}
Lastly, the ancilla blocks are measured in the occupation number basis i.e. $i\Gamma_1\bar{\Gamma}_1$ and $i\Gamma_2\bar{\Gamma}_2$ basis. Note that the parity operators i.e. $\Gamma_a\bar{\Gamma}_a, \Gamma_b\bar{\Gamma}_b$ on the data blocks commute with this measurement, whereas the primed operators do not, but similar to the previous examples, they can be multiplied by the ancilla stabilizer to fix them. This gives us
 
\begin{align}
    \Gamma_a' &\rightarrow \Gamma_a \Gamma_b \Gamma_r \bar{\Gamma}_r\\
    \Gamma_b'&\rightarrow \Gamma_b \bar{\Gamma_r}.
\end{align}
 We can see here that had we chosen two commuting operators, then we would have got an actual logical CNOT. But since they are fermionic in nature, and they all mutually anti-commute, instead we get a 
 \begin{align}
    \Gamma_a' \rightarrow \Gamma_a' \Gamma_b' \Gamma_r \bar{\Gamma}_r\\
    \Gamma_b'\rightarrow \Gamma_b' \Gamma_r \bar{\Gamma}_r
\end{align}
which corresponds to the following logical gate

 \begin{align}
\mathrm{CNOT}(1,2)\mathrm{CZ}(a,r)\mathrm{CZ}(b,r) 
 \end{align}

\section{Fault-tolerant error correction: Steane error correction gadget}\label{section:Steane}

In qubit CSS codes, transversal CNOT and transversal measurements allow us to use the Steane error correction gadget, which is fault-tolerant. As opposed to Shor, Steane measurement does not require repeated rounds of syndrome extraction, hence, for Majorana codes, it would be useful to have a similar scheme. 
Note that in Majorana systems, we can only measure Majorana bilinears. For a code on $2n$ Majorana modes, the size of the Majorana group $\mathrm{Maj}(2n)$ is $2^{2n+1}$. However, the size of the isotropic space is $2^n$, and thus one needs maximally $n$ generators that correspond to an Abelian subgroup of $\mathrm{Maj}(2n)$. This implies that one can make at most $n$ simultaneously commuting measurements.

Note that unfortunately, physical CNOTs are not allowed, however, we can use the $\braid$ gates to design a gadget. See figure \ref{fig:steanegadget}.
\begin{figure}
    \centering
    \includegraphics[trim={5cm 20cm 11cm 4cm}, clip=true, width=0.6\linewidth]{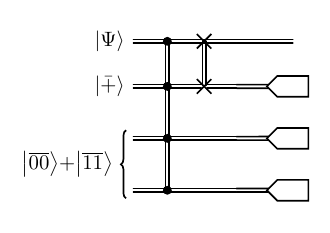}
    \caption{Steane gadget for four blocks of $[[n,1,d]]_f$ code., with transversal measurement in the occupation basis. This corresponds to correcting for $\gamma_i$ errors. The measurement blocks correspond to measuring in the occupation number basis, followed by classical post processing. For $\bar{\gamma}_i$, we conjugate this circuit with $\mathcal{B}$ given in the main text.}
    \label{fig:steanegadget}
\end{figure}

The gadget has support over four code blocks, where the first code block corresponds to the data block, while the remaining three correspond to ancilla blocks. Similar to Steane EC, we need to choose an ancilla state so that we can find error syndromes transversally, without measuring the logical information. We encode the three ancilla blocks in the same code as the data block. Moreover, we choose the following additional stabilizer conditions
\begin{itemize}
    \item $\bar{X_2}$ is inside the stabilizer, hence the second logical block is in $\ket{\bar{+}}$.
    \item $\bar{X_3}\bar{X_4}$ and $\bar{Z_3}\bar{Z_4}$ are also in the stabilizer, so the last two blocks are in a \textit{logical Bell} pair $\frac{1}{2}(\ket{\overline{00}}+\ket{\overline{11}})$.
\end{itemize}
We can correct errors of $\gamma_i$ type using this circuit. Note that
\begin{itemize}
    \item Errors on the data block propagate to ancilla blocks three and four, and also act on data block, at the same location.
    \item Errors on ancilla block two propagate to errors on block two, three and four.
    \item Errors on ancilla block three propagate to errors on the data block, block two and block four.
    \item Errors on ancilla block four propagate to errors on the data block, block two and block three.
\end{itemize}
Since measuring the syndrome on a single block does not uniquely identify the location of the fault we need syndromes from all four blocks. 
To represent the codewords of Majorana codes, we need to first select a complete set of basis states. For $2n$ Majoranas, one has $n$ physical fermions, with one identifying the pair $\gamma,\bar{\gamma}$ with each mode. The occupation number basis $\ket{n_1,n_2,n_3,...n_n}$ is the analogous computational basis in this context, and the default measurement basis. The codewords of this code are:
\begin{align}
    \ket{\bar{k}} = \prod_{i} (I + g_i)N^{k}\overbrace{\ket{00\ldots }
    }^n
\end{align}
where $g_i$ are rows of matrix $S$. Upto phases, we can multiply them to find the following form:
\begin{align}
  \ket{\bar{k}} = \sum_i g_i N^k \ket{00\ldots 0}  \\
  \ket{\bar{k}} = \sum_{v\in C} \ket{\bar{k}+v}
\end{align}
Here, $\bar{k}$ labels the coset number. For the second code block, the states can be written as
\begin{align}
\ket{\bar{+}} = \sum_{v\in C} \ket{\bar{0}+v}  + \sum_{v\in C} \ket{\bar{1}+v}
\end{align}
while for the third and fourth code blocks the states are:
\begin{align}
 \ket{\overline{++}} = \sum_{v,w\in C} \ket{\bar{0}+v,\bar{0}+w}  + \sum_{v,w\in C} \ket{\bar{1}+v,\bar{1}
 +w}
\end{align}  

One can measure in the occupation number basis, and use the parity check matrix of the classical code $C$ to check for low weight $\gamma$ errors. In order to correct the errors, we need to apply a recovery operation. By doing classical decoding on each block, it gives out a recovery operation that would cancel the syndrome. If the syndrome in each block is labeled  $s_j,\{j=2,3,4\}$, then the syndrome for the data block is given as
\begin{align}\label{eq:boolean}
    s_d = s_2\wedge s_3 +s_2\wedge s_4 +s_3\wedge s_4+s_2\wedge s_3\wedge s_4 \mod 2
\end{align}
This tells us that to identify the data syndrome, we need to combine information from the remaining syndromes, since for instance, an error on the data block raises a flag on same syndrome bits of blocks three and four, but nothing on block two, and thus the above function will evaluate the syndrome vector on the data block, which can then be used to estimate the error vector on the data and apply the corresponding correction. Moreover, since the data and ancilla blocks have error patterns, one can distinguish if the error came from a data block or ancilla block. More details about the syndrome recovery are in the Appendix \ref{section: syndromedetail}.

For errors of type $\bar{\gamma_i}$, we can conjugate with the transversal transformation:
\begin{align}
    \mathcal{B} = \prod_{i=1}^{n}\braidtwo(i,i+n) = \prod_{i=1}^{n}\exp(-\frac{\pi}{4} \gamma_i \bar{\gamma_i})
\end{align}
For $\gamma_i$ errors, we prepare the second ancilla block in the logical $\ket{\bar{+}}$ block, and measure transversally in computational basis. Instead if we have $\bar{\gamma}_i$ errors, we prepare in $\mathcal{B}\ket{\bar{+}}$ and measure in $\mathcal{B^{\dagger}}\ket{\text{occupation basis}}$, and since $\braidtwo \gamma_i \braidtwo  = \bar{\gamma}_i$, this transformation is sufficient to find syndromes of $\bar{\gamma}_i$. 

\textit{Recovery operation: }A recovery operation as described before corresponds to an odd weight operation on the data register. To do so, we can also keep an ancilla register such that the quantum operation on the ancilla+system is an even weight operation and does not violate parity superselection. We can initialize an ancilla system with $n_r$ modes, such that they are initially in a global even state. We can also keep an extra classical bit that stores the parity of the ancilla register. This scheme corresponds to a parity bit error detecting code and can only detect an odd number of parity errors.  A recovery operation thus corresponds to acting with $\mathcal{R\circ A}(\rho_{ad})$ where $\mathcal{A}$ is defined as follows
\begin{align}
    \mathcal{A}(\rho_{ad}) = \sum_{j=1}^{n_r} (1-\eta_{j+1})\gamma_j^a \eta_{j-1} (\rho_{ad}) \sum_{j=1}^{n_r} (1-\eta_{j+1})\gamma_j^a \eta_{j-1} 
\end{align}
where $\eta_j = i\gamma_j^a \bar{\gamma_j}^a $ which basically removes a Majorana from the last mode. We can also update the parity bit when we do so. In case there is a parity flip on the ancilla, we can measure $\Ptot$ of the ancilla register and compare against the stored classical bit, if those are different, then we have incurred a parity flip and one can reset this register to the prior state. Note that throughout this recovery operation, the ancilla and data registers remain unentangled: in a correct stabilizer subspace, both the ancilla and the data are even and can be written in a tensor product fashion. In the event of an odd error, and thereafter recovery, the code state is returned to its previous correct subspace that corresponds to an even code, while the ancilla is in an odd parity subspace. 

\subsection{Logical measurement in Steane}
To measure logical operators of the data block, we need to prepare the ancilla block in a state such that the logical state of the data is copied to it and can be transversally measured. In qubits, to measure $\bar{Z}_d$, ancilla block is prepared in the logical $\bar{0}$ state, or an eigenstate of $\bar{Z}$ operator. A transversal $\mathrm{CNOT}$ then copies the state of the data block: if there is a logical bit flip on the data, then the state of the ancilla block is also flipped, and a logical $\bar{Z}$ measurement thus detects this flip. We can use a similar idea for the logical measurement here.

To measure $\bar{Z}$ on the data block, we need to copy the logical operator from the data block to the ancilla blocks. As mentioned above, 
\begin{align}
    \bar{X}_d \rightarrow \bar{X}_d\bar{X}_3 \bar{X}_4
\end{align}
By preparing ancilla blocks three and four in eigenstate of $\bar{Z}$ for each respective block, we can measure each block transversally, and perform classical error correction as described above to find the closest logical state. A similar idea is present in order to perform a logical $X$ measurement, except now the ancilla state is prepared in an eigenstate of $\bar{X}$ rather than $\bar{Z}$. 

This protocol is important in two contexts, first we can combine a non-fault tolerant encoding scheme as in \cite{MudassarEncoding} with this gadget, and by measuring the stabilizers of the code and the logicals, we can verify that we prepared the correct state fault tolerantly and in case we detect any faults, we can discard and start again. Secondly, we can insert this gadget into various error locations and measure the stabilizer of the code to flush out any errors in the system. 

\section{Transversal non Clifford gates}\label{section:nonClifford}
While Clifford gates and measurements are a useful set of gates since they correspond to stabilizer circuits, in order to do universal computing, one also needs non-Clifford resources. $T$ gate is one instance, that along with fermionic Clifford operations, constitutes universal fermionic computing \cite{BRAVYI2002210}. In most schemes performing magic state distillation, one requires a code that contains a logical $T$ gate. To this end, we show an example of a Majorana code that implements a transversal logical $T$ gate.

 In the qubit quantum Reed Muller code, the codewords are associated to the cosets of $C_2^{\perp}$ in $C_1$. Note that each element in $C_2^{\perp}$ satisfies the \textit{triorthogonality} condition, and in particular for $[[15,1,3]]$, each basis state in the $\ket{\bar{0}}$ has weight $0\mod8$ 
\begin{align}
\ket{\bar{0}} = \sum_{v\in C_2^{\perp}} \ket{v}
\end{align}
where $C_2^{\perp}$ is generated by $H_2$
\begin{align}
H_2=
 \begin{bmatrix}
1 & 1 & 1 & 1 & 1 & 1 & 1 & 1 & 0 & 0 & 0 & 0 & 0 & 0 & 0\\   
1 & 1 & 1 & 1 & 0 & 0 & 0 & 0 & 1 & 1 & 1 & 1 & 0 & 0 & 0\\ 
1 & 1 & 0 & 0 & 1 & 1 & 0 & 0 & 1 & 1 & 0 & 0 & 1 & 1 & 0\\ 
1 & 0 & 1 & 0 & 1 & 0 & 1 & 0 & 1 & 0 & 1 & 0 & 1 & 0 & 1
\end{bmatrix}   
\end{align}

As a result, applying a rotation by $R^{\otimes n}_{\pi/8}$, where $R_{\pi/8}=e^{-i\frac{\pi}{8} \sigma_z }$, gives an overall global phase of $e^{i \phi}$, where $\phi=-n\pi/8$. 

Analogously, consider the $T$ gate for Majoranas:
\begin{align}
R_{\pi/8} = e^{-\frac{\pi}{8} \gamma_i\bar{\gamma_i}} = e^{-i\frac{\pi}{8} (-i\gamma_i\bar{\gamma_i})}
\end{align}

We discussed before that in general, for Majorana CSS codes, $C_2^{\perp}$ is not necessarily contained in $C_1$. However, if one starts with a qubit CSS code such that $C_2^{\perp}\subseteq C_1$, and chooses the parity check matrices $H_\gamma = H_x$ and $H_{\gamma \bar{\gamma}}=H_z$, then the corresponding Majorana code would have the required property. In other words, $H_{\gamma}$ generates $C_2^{\perp}$. This allows us to write codewords, for instance, for the $[[15,1,3]]_f$ code:
\begin{align}
    \ket{\bar{0}} = \sum_{v\in C_2^{\perp}} \ket{v}\\
    \ket{\bar{1}} = \sum_{v\in C_2^{\perp}} \ket{1+v}
\end{align}
where $\ket{1}\in C_1\setminus C_2^{\perp}$ and corresponds to all ones vector. More precisely, $C_1$ corresponds to punctured $\mathcal{R}(1,4)$ (last index is deleted) and $C_2^{\perp}$ corresponds to the even subcode, with the all ones vector removed. 

From the same argument, the $\ket{\bar{0}}$ has weight $0\mod8$, and by applying fermionic $(R_{\pi/8})^n$ gives $\phi = -n\pi/8$. The $\ket{\bar{1}}$ state has weight $7\mod 8$.
Applying the $(R_{\pi/8})^n$ yields:

\begin{align}
\prod_{i}^{n} R_{\pi/8} \ket{\bar{1}}
&=e^{-i(n-7)\pi/8} e^{i 7\pi/8} \ket{\bar{1}}\\
&=e^{-i n\pi/8} e^{i7\pi/4} \ket{\bar{1}}\\
&=e^{-i n\pi/8} e^{-i\pi/4} \ket{\bar{1}}
\end{align}
When applied to a superposition of $\ket{\bar{0}}$ and $\ket{\bar{1}}$ states, this results in
\begin{align}
    \prod_{i}^n R_{\pi/8} (\ket{\bar{0}}+\ket{\bar{1}}) = e^{i\pi/8}\ket{\bar{0}}+e^{-i\pi/8}\ket{\bar{1}}
\end{align}
which just shows that transversal $(R_{\pi/8})^n$ is indeed a logical gate of this code that corresponds to a logical $R_{-\pi/8}$ gate. Note that the logical operators of this code are odd, which implies that the above state can only be described with respect to a reference. 

Generically, triorthogonal codes are used to define codes with a transversal $T$ gate \cite{Haah_Low-overhead}. In the construction of such codes, the odd weight rows of a triorthogonal matrix are used to define logical qubits, and the odd weight is necessary for the correct anticommutation relations $\{\bar{X_i},\bar{Z_i}\}=0$. Here, we are promoting the binary matrix to correspond to Majorana CSS codes via Lemma ~\ref{lem_csstomaj}, so we have a similar constraint to have odd weight binary strings for the logical operators, since $\bar{X}$ and $\bar{Z}$ are chosen to have support on $\gamma$ part and $\bar{\gamma}$ respectively.  

This code is also useful to implement universal fault-tolerant gates from using code-switching \cite{Codeswitching1,Codeswitching2}, since both $[[7,1,3]]$ and $[[15,1,3]]$ correspond to Majorana Steane code $[[7,1,3]]_f$ and Majorana Reed Muller code$[[15,1,3]]_f$ , where Steane code implements Clifford gates and Reed Muller implements $T$ gate. Switching can be implemented by measuring stabilizers of the respective code, and applying gauge operators if a correction is required. Since gauge operators correspond to faces of the code and are even weight, they can be measured without the need of a reference. While we do not show this, it is likely that higher distance 3D color codes can also be converted to Majorana Reed Muller codes, since such codes are triorthogonal and automatically satisfy the self-orthogonality property required for Majorana codes. 

\subsection{Magic state distillation using this code}
We describe the procedure to implement magic state distillation using the $[[15,1,3]]_f$ fermion code. A magic state distillation procedure starts with an input of noisy magic states. Note that a fermion magic state $\ket{T_f}$:
\begin{align}
    \ket{T_f}=R_{\pi/8}\ket{+} = \ket{0}+ e^{i\pi/4}\ket{1}
\end{align}
corresponds to a coherent superposition of even and odd states, and thus we need a reference state. Therefore, we have
\begin{align}
    \ket{T_f}_{\text{inv}} = R_{\pi/8}(\ket{0}_r\ket{+})
\end{align}
instead. For a non-fault-tolerant scenario, we can start with 15 noisy T states. 
\begin{align}
    \ket{T_f}_{\text{inv}} = R_{\pi/8}^{\otimes 15}\ket{0}_r \ket{+}^{\otimes 15} 
\end{align}
which can be easily extended to be fault tolerant by encoding these in an outer code $Q$. Note that in order to use the fermionic $[[15,1,3]]_f$, this outer code must be odd as well, so that the logicals are fermions. Thus, both the reference and the magic states are encoded in an odd fermion code, $Q$. One can then do FTEC to make this encoding procedure fault tolerant.

Next, one can measure the projectors of the inner [[15,1,3]] code $\Pi_{15}$. 
\begin{align}
    \Pi_{15} \overline{R}_{\pi/8}^{\otimes 15} \ket{\bar{+}}^{\otimes 15}
\end{align}
The projectors commute with $R_{\pi/8}^{\otimes 15}$ since it is transversal, and project $\ket{\bar{+}}^{\otimes 15}$ to $\overline{\ket{\bar{+}}}$ since it is a Majorana CSS code (the only surviving codewords are the ones dual to $C_2$). If there are any weight-1 or weight-2 errors on the code, they can be detected and the code discarded. Else, we can decode the inner code to return an encoded $\ket{T}_f$ with better fidelity.  
As alluded to before, one can also use a qubit [[15,1,3]] code, but the overhead would be bigger by a factor of two.

\section{Majorana CSS codes}\label{section: CSScodes}

In \cite{bravyi_majorana_2010}, a qubit weakly self-dual CSS code is converted to a Majorana CSS code i.e. one code specifies the support of $\gamma_i$, while the other specifies the support of $\bar{\gamma_i}$. In such a construction, since qubit codes have to satisfy the constraint that $H_xH_z^{T}=0\mod 2$, codes for both $\gamma$ and $\bar{\gamma}$ are same. However, for Majorana codes this is actually not required in general, so we extend this construction here. 
Consider two classical codes $C_1$, $C_2$, and their generator matrices correspond to $G_1$ and $G_2$ respectively. The Majorana CSS construction is given as
\begin{align}
    S = \begin{bmatrix}
        G_1 & 0\\
        0 & G_2
    \end{bmatrix},
\end{align}
For two binary vectors $\textbf{x}$ and $\textbf{z}$ drawn from $G_1$ and $G_2$ respectively, we have
\begin{align}
\textbf{x} \Lambda_f \textbf{z}^T = 0.
\end{align}

The evaluation of the above gives
\begin{align}
    \sum_{i=1}^N x_i \sum_{j=1}^N z_j = 0. 
    \label{eq:csscond1}
\end{align}
This implies that every vector $\textbf{x}$, \textbf{z} must have even weight.
Elements of stabilizers from the same classical code also should commute, therefore we require that two vectors $\textbf{x}_1$ and $\textbf{x}_2$ from the same classical code fulfill
\begin{align}
\textbf{x}_1 \Lambda_f \textbf{x}_2^T = 0,
 \end{align}
 which gives
 \begin{align}
     \sum_{i=1,j\neq i}^N x_{1i} x_{2j}=0.
     \label{eq:CSS2}
 \end{align}
  Eq.\eqref{eq:CSS2}, combined with Eq.\eqref{eq:csscond1}, which states that each row is even weight, which gives
 \begin{align}
     \sum_{i=1} x_{1i} x_{2i} = 0,
\end{align}
which is
\begin{align}
    \textbf{x}_1\cdot \textbf{x}_2=0.
\end{align}
This implies that both codes $C_1$ and $C_2$ should be weakly self dual i.e. $C_1\subseteq C_1^{\perp}, C_2 \subseteq C_2^{\perp}$. 
One can thus use any two classical $[n,k_1,d_1]$ and $[n,k_2,d_2]$ self orthogonal codes to construct a Majorana CSS code, $[[n,n-(k_1+k_2),\min(d_1,d_2)]]_f$. In general, for degenerate Majorana codes, $d\geq \min(d_1,d_2) $. Degeneracy is possible for instance for high distance Majorana color codes, where each stabilizer has weight six, but distance can be varied macroscopically and can be scaled as $l$. For such a case, product of two low weight detectable errors $E$ and $E'$ can be inside $S$.   

The Majorana LDPC code obtained below is constructed through Lemma~\ref{lem_csstomaj}, however in general, one can expect to have an LDPC code construction that can be a Majorana CSS code but not a qubit CSS code. Another important observation in this construction is that if $\Ptot\in S$, then one cannot form logical fermions since in a CSS code, logical operators can be either of type $\gamma$ or $\bar{\gamma}$ but not both, and in order for all of them to mutually anti-commute, they must be have operators of type $\gamma$ and $\bar{\gamma}$.  
\section{Majorana quantum LDPC code}\label{section: qldpc code}
For a Majorana code, a high-rate encoding would be beneficial when encoding multiple logical qubits or fermions while having a high distance. Note that we have the following constraint on the stabilizer of Majorana codes
\begin{align}
    SS^{T} = 0.
\end{align}
This precludes quantum LDPC codes generated through well known product constructions, such as the hypergraph product. However, as described in the previous section, one can choose a classical self-orthogonal code, and that would correspond to a Majorana LDPC code. 

One such construction is given in \cite{couvreur2013constructionquantumldpccodes} which gives a classical weakly self-dual code, and can be converted to a quantum weakly self dual code via \cite{Steane96}. One can then use Lemma ~\ref{lem_csstomaj} to convert to a Majorana code. We now summarize the construction and the key results.

The construction involves the use of \textbf{Cayley graphs}. Let there be a group $G$ and a subset of that group $S$. A Cayley graph $\mathscr{C}(G,S)$ is defined as a graph whose vertex set is $G$, and there exists an edge between two vertices $g,g'$ iff there exists $s\in S$ such that $g=sg'$. 

If one chooses $S$ such that $S$ generates $G$, and $S$ has the following properties:
\begin{itemize}
    \item The number of generators in $S$ are even.
    \item For all $g\in G$, there are an even number of distinct expressions of $g$ s.t. $g=st^{-1}$, where $s,t\in S$.
\end{itemize}
Then the adjacency matrix $\mathscr{R}$ of the Cayley graph has the following properties:
\begin{itemize}
    \item Each row has even weight, since each element in $G$ is connected to an even number of vertices.
    \item Each row is orthogonal to any other row.  Consider two elements $g_1,g_2\in G$. In the adjacency matrix, they have a 1 in common iff $tg_1=sg_2$ i.e. they are connected to the same element. This is equivalent to $g_1g_2^{-1} = st^{-1}$. Since $g_1g_2^{-1}\in G$, and there are an even number of expressions for any group element, the overlap is even.

\end{itemize}
The adjacency matrix $\mathscr{R}$ now has the property of self orthogonality. This generates a classical code $C$ such that:
\begin{align}
    C \subseteq C^{\perp}
\end{align}.

To generate a Majorana code, we can either:
\begin{itemize}
    \item Promote the generator matrix of this code to the stabilizer matrix of the Majorana code. By construction, this will have the following parameters:
    if the classical code $C$ has parameters $(2n,k,d)$, then the Majorana code will have $[[n,n-k,d^{\perp}]]$, where $d^{\perp}$ is the dual distance given by the weight of smallest element in $C^{\perp}$. 
    \item Let $H_1=\mathscr{R}$ and $H_2=\mathscr{R}$. Now, the parameters, given a classical code $[n,k,d]$(where $\mathscr{R}$ is generator matrix for the classical code) are $[[n,n-2k,D^*]]$, where $D^*\geq d^{\perp}$. 
\end{itemize}

For our purpose, we use the second method, which generates a Majorana CSS code.

In general, $k$ and $d$ of the classical code are given as:
\begin{align}
    k = \dim(C)=\rank(\mathscr{R})\\
    d = \min(\text{wt}(C\setminus C^{\perp})) = \min(\text{wt}(\Ima{(\mathscr{R}})\setminus \ker({\mathscr{R}})
    )
    )
\end{align}
while for the Majorana code, the parameters are:
\begin{align}
    K = \dim(C)=N-2\rank(\mathscr{R})\\
    D^{*}=\min(\text{wt}(C^{\perp}\setminus C)) = \min(\text{wt}(\ker{(\mathscr{R}})\setminus \Ima({\mathscr{R}})))
\end{align}
Note that $d^{\perp}$ is too loose of a lower bound for distance of the code
since the rows of generator matrix have low weight by design, therefore, the distance is computed directly. 
Finding these is non-trivial in general, but \cite{couvreur2013constructionquantumldpccodes} computes these for the case of the repetition code. The construction is as follows:
\begin{itemize}
    \item Consider the parity check matrix $H_n$ for a repetition code $[n+1,1,n+1]$, with size $n\times (n+1)$. The columns of this matrix correspond to $S$, while the group $G=\mathbb{F}_2^{n}$ i.e. all bitstrings of length $n$. This gives a vertex set of size $2^n$.  Note that the the size of $S$ must be even,  so a non-trivial LDPC code only exists for repetition codes with odd $n$. 
    \item Find the adjacency matrix corresponding to the above, represented as $\mathscr{R}(S_n,\mathbb{F}_2^n)$. 
\end{itemize}
For odd $n$, and $n\geq 3$,  the rank is given as:
\begin{align}
    \rank(\mathscr{R}) &= N-\ker(\mathscr{R})\\
    &=2^n - (2^{n-1}+2^{\frac{n-1}{2}})
\end{align}
Finally, the Majorana code parameters are:
\begin{align}
    K = N-2(\rank(\mathscr{R}))=2^{\frac{n+1}{2}}\\
    D^{*} = \min(\wt(\ker(\mathscr{R})\setminus \Ima{\mathscr{R}})) = 2^{\frac{n-1}{2}}
\end{align}
 Thus we have found a $[N,\sqrt{2N},\sqrt{N/2}]$ code, which implies it is high rate and high distance, in fact, the distance is comparable to a hypergraph product (HGP) construction \cite{Zemor}, but $k$ is $O(\sqrt{N}$ as opposed to $O(N)$ for HGP codes. Moreover, check weights also grow as $O(\log(N))$, while for HGP codes, the checks are normally asssumed to be constant weight.  

One important question is whether this code encodes logical qubits or logical Majoranas. One method to check would be if $\Ptot\in S$. If that is the case, then all logical operators can only be \textbf{even} weight. We show that this is true in the following theorem. 
\begin{theorem}
    A Majorana code generated from the construction in Ref.~\cite{couvreur2013constructionquantumldpccodes} using the repetition code as the underlying classical code has only even weight logical operators. 
\end{theorem}

\begin{proof}
In order to test whether a code has even or odd logical operators, we need to test whether the total parity is contained inside the stabilizer or not. We can look at two different cases, when $n$ is odd and when $n$ is even.

Once again, we have a Cayley graph $C(G,S)$, and the stabilizer matrix of our desired LDPC code corresponds to its adjacency matrix. For the repetition code, the parity check matrix has the form
\begin{align}
    H_n =[e_1,e_2, \ldots ,e_n, e_1\oplus e_2\oplus\ldots e_n], 
\end{align}
where checks correspond to the rows of $H_n$. $e_i$ are the column vectors of length $n$ and correspond to the canonical basis vectors of $F_2^n$. The columns of $H_n$ form the subset $S$ while $G=\mathbb{F}_2^n$. Note that for $n+1$, $H_{n+1}$ is the same except with an additional $e_{n+1}$ vector added to above.The adjacency matrix is defined $A_{n+1}\equiv \mathscr{R}(F_2^n,S_n)$ and in this case, it has the following nice recursive form because of the structure of the parity checks of repetition code
\begin{align}
A_{n+1}=
\begin{bmatrix}
A_{n}+J_{n} & I_{n}+J_{n}
\\
I_{n}+J_{n} & A_{n}+J_{n}
\end{bmatrix}
\end{align}
where $J_{n}$ is the binary antidiagonal matrix,
\begin{align}
    {[J_n]}_{(i,2^{n-1}-i)} = 1 \quad \text{for } i=\{0,1\ldots 2^{n-1}\}
\end{align}
and $I_{n}$ is the identity matrix, each of dimension $2^{n-1}$.  To see why one gets the above form, note the following:

     Each element \textbf{x} in $\mathbb{F}_2^n$ can be mapped to an integer \textbf{p} in range $\{0,\ldots 2^n-1\}$, corresponding to its lexicographic ordering where \textbf{x} is the binary representation of \textbf{p}. For instance, 
    \begin{center}
    \begin{tabular}{|c|c|}
    \hline
        \textbf{x} & \textbf{p}\\
    \hline   
    $\overbrace{00\ldots 0}^{n}$ & $0$\\
    $00\ldots 1$ & $1$\\
    \vdots & \vdots\\
    $11\ldots 1$ & $2^{n}-1$\\
    \hline
    \end{tabular}
    \end{center}
    $\textbf{p}$ indexes the position in the adjacency matrix $A_{n+1}$.
    
    Note that the top left block of $A_{n+1}$ corresponds to a binary antidiagonal matrix . This is because $\exists x\in \mathbb{F}_2^n$ such that we can map to another element $y \in \mathbb{F}_2^n$ by multiplying with $s\in S_n$ where $s$ is the all ones vector. This basically swaps 0 and 1 in $x$. What this means is that:
    \begin{align}
        [A_{n+1}]_{(p,2^n-1-p)} = 1
    \end{align}
    Hence we have $J_n$ in the top right half and bottom left half of $A_{n+1}$.
     Consider an element $x\in \mathbb{F_2^n}$ whose $n^{th}$ position from the right is 0. This corresponds to a $\textbf{p}$ in range $\{0,\ldots 2^{n-1}-1\}$. Adding an element $e_n$ maps $p$ to $p+2^n$, which gives $I$ in top right block, and symmetrically in the bottom left block.   
     
    Lastly,  we need to remove the contribution of $e_1 +e_2+\ldots e_n$ from $A_n$ so we add $J_n$ to the top left matrix, which as described above, corresponds to adding an all one vector.

For each column in $A_{n+1}$, the number of ones are given by:
\begin{align}
(\text{\# of ones per col } A_{n+1} )= (\text{\# of ones per col } A_{n}) +1 
\end{align}
since for each column, the identity matrix contributes a single one, while the ones in $J_n$ matrices cancel each other.    
In addition, $A_4$ has four ones in each column, so in general, $A_{n+1}$ has $n+1$ ones in each column.

For odd $n+1$, showing that $\Ptot$ is inside the group is trivial since each column has an odd number of ones, so adding all the rows $\mod{2}$ in general will give $\Ptot$. Odd $n+1$ also has trivial logical dimension. 

For even $n+1$, let's look at the top half of the matrix
\begin{align}
    \begin{bmatrix}
        A_{n}+J_{n} & I_{n} +J_{n}
    \end{bmatrix}
\end{align}
The size of the above matrix is $2^{n-1}\times 2^{n}$, hence each submatrix is of size $2^{n-1}$. Represent the left submatrix by $L$ and the right submatrix by $R$. 

Note now that the sum of the top half of $R$ corresponds to the all ones vector of size $2^{n-1}$. Thus the sum of the first $2^{n-2}$ rows of $[0|R]$ give the vector

\begin{align}
\begin{bmatrix}
    00\cdots 0 |1\cdots 1
\end{bmatrix}
\end{align}
of length $2^{n}$, and the ones part is $2^{n-1}$.
Let us now look at the top half of the left submatrix $L$. 

\begin{align}
    L = \begin{bmatrix}
        A_{n}+J_{n}
    \end{bmatrix}
    \label{eq:parityright}
\end{align}
The sum of the all rows of top half (or the first $2^{n-2}$ rows) of $J_{n}$ gives 
\begin{align}
\begin{bmatrix}
00\cdots 0 , 11\cdots1
\label{eq:parityleft1}
\end{bmatrix}  
\end{align}
of length $2^{n-1}$, where the ones part is of length $2^{n-2}$.  The comma separation indicates that we are looking at only the left submatrix. Let us look at the contribution from the top half of $A_{n}$. Note that top half of $A_{n}$ is:
\begin{align}
\begin{bmatrix}
  A_{n-1}+J_{n-1} | I_{n-1}+J_{n-1}   
\end{bmatrix}
\end{align}
The right part of the above matrix is all zeros. Since $A_{n}$ corresponds to the odd $n+1$ case, from the previous argument we can conclude the sum of bottom left half of $A_{n}$ gives 0, the sum of left columns gives 
\begin{align}
\begin{bmatrix}
    11\ldots 1 , 00\ldots0
    \label{eq:parityleft2}
\end{bmatrix}   
\end{align}
of length $2^{n-1}$, where the ones vector is of size $2^{n-2}$. Combining \ref{eq:parityleft2} and \ref{eq:parityleft1} gives all ones in the right submatrix. Combining this with \ref{eq:parityright} gives the total parity $\Ptot$ inside the stabilizer, and concludes the proof.
 
\end{proof}
To look at an example for the sake of illustration, let us look at the $n+1=4$ case. The adjacency matrix is given as:
\begin{align}
A_4=
\begin{bmatrix}
0 & 1 & 1 & 0 & 1 & 0 & 0 & 1\\ 1& 0& 0& 1& 0& 1& 1& 0\\ 
1& 0& 0& 1& 0& 
  1& 1& 0\\ 
  0& 1& 1& 0& 1& 0& 0& 1\\
  1& 0& 0& 1& 0& 1& 1& 0\\
  0& 1& 1& 0& 1& 0& 0& 1\\
   0& 1& 1& 0& 1& 0& 0& 1 \\ 
   1& 0& 0& 1& 0& 1& 1& 
  0
\end{bmatrix}
\end{align}
The parameters for this code are: $N=8,K=2^2=4,D^* =2 $.  
 Also each row has an even overlap with any other row, which gives it the self orthogonal property. Also note that the rank of this matrix is two. We can use the CSS construction for Majoranas to obtain a Majorana CSS code with the following operators
\begin{align}
S_1 = \gamma_2 \gamma_3 \gamma_5 \gamma_8\\
S_2 = \gamma_1 \gamma_4 \gamma_6 \gamma_7\\
S_3 = \bar{\gamma}_2 \bar{\gamma}_3 \bar{\gamma}_5 \bar{\gamma}_8\\
S_4 =\bar{\gamma}_1 \bar{\gamma}_4 \bar{\gamma}_6 \bar{\gamma}_7
\end{align}
There are four logical operators, which can be found by computing the kernel of the stabilizer matrix:
\begin{align}
\bar{X_1} = \gamma_2 \gamma_3 \qquad
\bar{Z_1} = \gamma_2 \gamma_5\\
\bar{X_2} = \gamma_1 \gamma_4\qquad
\bar{Z_2} = \gamma_1 \gamma_6\\
\bar{X_3} = \bar{\gamma}_2 \bar{\gamma}_3 \qquad
\bar{Z_3} = \bar{\gamma}_2 \bar{\gamma}_5 \\
\bar{X_4} = \bar{\gamma}_1 \bar{\gamma}_4 \qquad
\bar{Z_4} = \bar{\gamma}_1 \bar{\gamma}_6 
\end{align}
This gives an $[[8,4,2]]_f$ code, with eight physical fermions. The above construction gives logical qubits, but one can also find logical fermions by choosing an anti-commuting basis, although such logicals may not be of a CSS type. We note that there has been some subsequent work on Majorana LDPC codes that encode logical fermions \cite{xu2025fermiontofermionlowdensityparitycheckcodes}. 
\section{Discussion and conclusion}
In this paper, we have given a comprehensive framework for fault tolerant gadgets in architectures with fermions as physical degrees of freedom, for both Clifford and non-Clifford gates, as well as measurements. We have demonstrated important differences between even and odd fermionic codes, and used a quantum reference system to enact gates in odd codes. We have also given a construction of Majorana CSS codes and Majorana LDPC codes that encode logical qubits.

Our framework opens up new questions. In particular, with an appropriate decoder, we argue that our fault-tolerant gadgets would enable estimation of thresholds under various noise models for different quantum codes, which would be interesting to explore. Moreover, we have introduced LDPC codes that encode logical qubits, and it would also be interesting to explore frameworks that encode logical fermions.  Lastly, for topological computing on Majorana devices, it would be beneficial to have a unitary implementation of the $\braid$ gate that we have used extensively. We leave these questions for future work. 
\section{Acknowledgements}
M.M acknowledges useful discussions with Yifan Hong, Riley Chien and Chris Fechisin. Both M.M. and D.G. are supported under RQS QLCI grant OMA-2120757. A.S. acknowledges support by the U.S.~Department of Energy, Office of Science, National Quantum Information Science Research Centers, Quantum Systems Accelerator (QSA).

\appendix
\section{Parity superselection as a constraint on operators}\label{section: parityconstraintunitary}
We explained in the main text why parity superselection obstructs certain operations. Here we give a concrete form for why unitaries (in the absence of reference frames) can only be physically realizable if they commute with the total parity.

Consider the density operator, which in general can take the form 
\begin{align}
    \rho = \rho_e \oplus \rho_o
\end{align}
Assume that the starting state is $\rho_{e}$ , suppose we perform a parity switching unitary $U$ such that we have $\rho_{o}$. We can represent a general unitary by the following form
\begin{align}
    U=\left[\begin{array}{c|c}
      U_{ee} & U_{eo}\\
      \hline
      U_{oe} & U_{oo}
    \end{array}
    \right]
\end{align}
Conjugating $\rho = \rho_{e}$, with $U$, it is clear that one gets $\rho_o$ corresponds  iff $U_{eo}$ is non-zero. However, parity SSR only allows $U$ to have a block diagonal structure, so we arrive at a contradiction. This means that we cannot implement parity switching operations unitarily. 

\section{Details on $\braid$ implementation}\label{AppendixB}
A Braiding transformation on a system of Ising anyons can be written in terms of $F$ and $R$ matrices acting on the Ising fusion space, and in particular:
\begin{align}
    \braidtwo = F R^2 F
\end{align}
However, for the $\braid$ case, it is likely that this corresponds to a fusion scenario, since $F$ and $R$ only correspond to basis transformations and do not change the number of outcomes. The simplest example of fusion corresponds to measurement, or projection operators, which are depicted as four mode projectors. In the aforementioned scheme, we showed that a $\braid$ can be written in terms of four mode projectors, while using an ancilla mode. One can ask if it can be simplified to just two mode projectors. We show that this is not possible, as follows.

Consider the starting state of the stabilizer group to be 
\begin{align}
    S = \langle M_1,M_2,\ldots \rangle
\end{align}
For the sake of simplicity, these can be taken to be weight-two operators. The goal of a $\braid$ is to increase the weight of some stabilizer upon conjugation such that for any choice of generators, there is at least one element with weight four. We can prove that two-mode measurements are not sufficient to do this. Assume that one measures an observable $N$, which anti-commutes with two generators (note that for the above choice of generators, it can anti-commute with at most two). Then the updated stabilizer is:
\begin{align}
    S= \langle M_1M_2,N,\ldots\rangle 
\end{align}
At this stage, multiplying $M_1M_2N$ gives an operator of weight two
\begin{align}
    |M_1M_2N| = |M_1M_2|+|N|-2|M_1M_2 \cap N |\\
    =|M_1|+|M_2|+|N|-4
    =|M_1|+|M_2|-2
\end{align}
So this choice has still not changed the distribution of generators to the correct form. 
We can pick another generator, e.g. $N'$ that anti-commutes with $N$, but not with $M_1M_2$, so as to remove $N$ from the stabilizer as the next measurement, however this will also overlap on an even number of sites with $M_1M_2$, so the distribution does not change. 

To see that $\braid$ operators generate the logical $\braid$ group, we can show that the generators satisfy the Yang-Baxter identity:
\begin{align}
    \sigma_i\sigma_{i+1}\sigma_i = \sigma_{i+1}\sigma_{i}\sigma_{i+1}
\end{align}
The above identity is satisfied for the following choice of generators:
\begin{align}
    \sigma_i = \braid(1,2,3,4)\\
    \sigma_{i+1} = \braid(2,3,4,5)\\
\end{align}
This is easy to check.
\begin{align}
\braid(1,2,3,4) = \frac{1}{\sqrt{2}}(1+i\gamma_1 \gamma_2 \gamma_3 \gamma_4) = \frac{1}{\sqrt{2}}(1+\Gamma_1 \Gamma_2)\\
\braid(2,3,4,5) = \frac{1}{\sqrt{2}}(1+i\gamma_2 \gamma_3 \gamma_4 \gamma_5) = \frac{1}{\sqrt{2}}(1+\Gamma_2 \Gamma_3)
\end{align}
where $\Gamma_1 = \gamma_1$, $\Gamma_2 = i \gamma_2 \gamma_3 \gamma_4$, $\Gamma_3 = \gamma_5$

\begin{align}
    \sigma_i\sigma_{i+1}\sigma_i= \frac{1}{2\sqrt{2}}(1+\Gamma_1\Gamma_2)(1+\Gamma_2\Gamma_3)(1+\Gamma_1\Gamma_2)
\end{align}
Define $\Gamma_1\Gamma_2 \equiv a_1$, $\Gamma_2\Gamma_3\equiv a_2$
\begin{align}
    (R\otimes I)(I\otimes R)(R\otimes I)=(1+a_1)(1+a_2)(1+a_1)\\
    (1+a_1+a_2+a_1a_2+a_1+a_1^2+a_2a_1+a_1a_2a_1)
\end{align}
It can be verified that 
\begin{align}
    a_1^2 = a_2^2=-1, \{a_1,a_2\}=0
\end{align}
so the expression above gives
\begin{align}
(R\otimes I)(I\otimes R)(R\otimes I) = 2a_1+2a_2
\end{align}
which is symmetric under exchange of $a_1$ and $a_2$, showing that the Yang Baxter identity is satisfied. 
For the braid group, we have a set of $n$ linearly independent Majorana operators and conjugation by a braid operator corresponds to a swap between $n$ Majorana basis vectors. Similarly, $\braid$ performs swaps between $2k$ logical Majorana operators $\Gamma_i$, which are defined as follows

\begin{align}
    \Gamma_{2i+1} = \gamma_{4i+1} \text{ for } i=\{0,1,\ldots k-1\} \\
    \Gamma_{2i+2} = i \gamma_{4i+2}\gamma_{4i+3}\gamma_{4i+4}
\end{align}
$\braid$ thus acts on logical Majorana operators as
\begin{align}
    \braid^1 (\Gamma_1)(\braid^1)^{\dagger} = \Gamma_2\\
    \braid^1 (\Gamma_2)(\braid^1)^{\dagger} = -\Gamma_1
\end{align}
And correspondingly
\begin{align}
\Gamma_1 \xleftrightarrow[\braid^1]{} \Gamma_2 \xleftrightarrow[\braid^2]{} \Gamma_3 \xleftrightarrow[\braid^3]{} \ldots \Gamma_{k-1}\xleftrightarrow[\braid^k]{} \Gamma_k
\end{align}
modulo phases, where $\braid$ gates are defined
\begin{align}
    \braid^{2i}= \braid(4i-2,4i-1,4i,4i+1)\\
    \braid^{2i-1}=\braid(4i-3,4i-2,4i-1,4i)
\end{align}
\section{Decoding even and odd codes and logical Jordan Wigner transformation}
It is argued in \cite{MudassarEncoding} that when $\Ptot\in S$, one needs an ancilla to find a decoding map. We justify this further for even codes and show that they can only be decoded partially.

Consider an even code, and one requires a decoding map $\mathcal{D}$ such that:
\begin{align}
    \mathcal{D}:S_{\text{enc}}\rightarrow i\gamma_i\bar{\gamma}_i \quad i=\{0,\ldots n-k\}
\end{align}
However, if such a map exists, then $\Ptot\notin S$, and one will have at least one logical operator which is odd weight. But since a decoding map can only be constructed from even weight Cliffords, it is not possible to change the parity of logical operators. This means that either the logical operators have support outside of data qubits, or that one can only partially decode such a code such that the stabilizer group after the partial decoding still consist of $\Ptot$.

Consider now an odd code, which has been decoded completely and the stabilizer elements only have support over the first $2r$ modes. Since the logical operators must commute with the stabilizer, they can correspond to even or trivial support on the stabilizer modes. Suppose I choose a basis of logical operators (here I number the $2r+1$ as 1 and so on) such that the first mode is an odd operator and thus corresponds to a logical fermion, while the rest are even weight but correspond to logical fermions. We show how a  \textbf{logical} Jordan Wigner transformation on the remaining modes gives logical qubits.
\begin{align}
    \Gamma_1  = \gamma_1\bar{\gamma}_1 \quad \bar{\Gamma}_1 = \gamma_1\\
    \Gamma_2 = \gamma_2 \gamma_3\quad \bar{\Gamma}_2 = \gamma_2 \bar{\gamma}_3\\
    \vdots\\
    \Gamma_k = \gamma_2\bar{\gamma_k}\quad \bar{\Gamma}_k = \gamma_2 \gamma_{k+1}
    \end{align}
then I can perform a logical Jordan-Wigner transformation to obtain qubit operators as follows
\begin{align}
    \Gamma_1 = \gamma_1\bar{\gamma_1} \quad \bar{\Gamma}_1 = \gamma_1\\
    \Gamma_2 = \gamma_2 \gamma_3 \quad \bar{\Gamma}_2 = \gamma_2\bar{\gamma}_3\\
    \vdots\\
    \Gamma_k = \gamma_2 \prod_{i=3} \gamma_i \bar{\gamma}_i \gamma_{k+1} \quad \Gamma_k = \gamma_2 \prod_{i=3} \gamma_i \bar{\gamma}_i \bar{\gamma}_{k+1}
\end{align}
We can convert Majoranas into qubits, using physical Jordan Wigner transformation
\begin{align}
    \Gamma_2 = Y_2X_3  \quad \bar{\Gamma}_2 = Y_2Y_3\\
    \vdots\\
    \Gamma_k = Y_2 X_{k+1} \quad \Gamma_k = Y_2 X_{k+1}
\end{align}
which represents the decoded qubits in the even parity sector.

\section{Encoding circuit with reference modes}\label{section:encodingcircuitforref}

In the main text we emphasized that for logical operators with odd weight, we need a reference state to violate parity superselection. This requires the preparation of a global parity invariant state. In other words, we have a parity constraint on the system and reference. A general parity invariant state on $n$ parties can be written as:
\begin{align}\label{parityinvstate}
    \ket{\psi}_{\text{inv}} = (\frac{1+\Ptot}{\sqrt{2^{n-1}}})\sum_{\{i_1,i_2,\ldots i_n=0,1\}}\ket{i_1,i_2,\ldots i_n}
\end{align}
where 
\begin{align}
    \ket{i_1,i_2,\ldots i_n}=\gamma_1^{i_1,}\gamma_2^{i_2}\ldots \gamma_n^{i_n}\ket{\text{vac}}
\end{align}
where one of the parties, say the first one, is the reference, and the rest are the system. 

For two modes, labeled $R$ and $A$, respectively denoting reference and Alice, the parity invariant state is
\begin{align}
    \ket{\psi}_{\text{inv}} = \frac{1}{\sqrt{2}}(\ket{00}_{RA}+\ket{11}_{RA})
\end{align}
To prepare the general state in \ref{parityinvstate}, one can use the circuit 
\begin{multline}
    C\ket{00\ldots 0}=\prod_{i=0}^{n-2} \braidtwo(1, n - i)^{\dagger}\ket{00\ldots 0} \\
    = \braidtwo(1, n)^{\dagger} \braidtwo(1, n - 1)^{\dagger} \cdots \braidtwo(1, 2)^{\dagger}\ket{00\ldots 0}
\end{multline}
For $n=2$, where first mode now refers to reference and second to Alice
This can be prepared using $\braidtwo$ gates, as follows:
\begin{align}
    \ket{\psi}_{\text{inv}}=\braidtwo(R,A)^{\dagger} (\ket{00})\\
    =\frac{\ket{00}_{RA}+\ket{11}_{RC}}{\sqrt{2}}
\end{align}

For $n=3$, as is given in the main text, 
\begin{align}
    \ket{\psi}_{\text{inv}}=\frac{1}{2}(\ket{000}_{RAC} + \ket{011}_{RAC} + \ket{101}_{RAC}+\ket{110}_{RAC})
\end{align}
can be found using 
\begin{align}
   \ket{\psi}_{\text{inv}} =\braidtwo^{\dagger}(R,C)\braidtwo^{\dagger}(R,A)\ket{000}
\end{align}

\section{Jordan Wigner as a linear map}\label{section:JordanWignermap}
Quantum channels correspond to linear maps that are CPTP. One question that someone may have would be whether a logical transformation between logical operators that behave as qubits and logical operators that behave as a fermions is a physical map at all? We show that it is not a physical map since it is not CP. 

An (inverse) Jordan Wigner transformation maps operators with qubit commutation relations to operators with fermionic anticommutation relations.  In particular
\begin{align}
    \Lambda(X_i) = \prod_{j<i}Z_j X_i\\
    \Lambda(Y_i) = \prod_{j<i} Z_j Y_i
\end{align}
where $\Lambda$ corresponds to the inverse Jordan Wigner map. It is clear that the input Pauli strings commute for different $i$, while output strings mutually anticommute, so a unitary transformation cannot actually do this. Next, to check if this map is CPTP, we can use the Pauli transfer matrix and Choi Jamiolkowski formalism. The Pauli transfer matrix form is:
\begin{align}
    (R_{\Lambda})_{i,j} = \frac{1}{d}\sum \Tr[P_i\Lambda(P_j)]
\end{align}
For the sake of simplicity, we can evaluate for a two qubit inverse Jordan Wigner transform. The matrix has the following expression
\begin{align}
    (R_{\Lambda})_{i,i} = 1 \quad i\notin \{4,5,7,8,10,11,13,14\}\\
    (R_{\Lambda})_{i,i} = 0 \quad i \in \{4,5,7,8,10,11,13,14\}\\
    (R_{\Lambda})_{4,7} = 
    (R_{\Lambda})_{7,4} =  
    (R_{\Lambda})_{5,8} = 
    (R_{\Lambda})_{8,5} = 1 \\
    (R_{\Lambda})_{10,13} = 
    (R_{\Lambda})_{13,10} = -
    (R_{\Lambda})_{14,11} = 
    (R_{\Lambda})_{11,14} = -i 
\end{align}
One can then compute Choi matrix corresponding to this
\begin{align}
    \rho_{ij} = \frac{1}{d^2}\sum_{i,j} (R_{\Lambda})_{i,j} P_j^T \otimes P_i
\end{align}
This can be evaluated to find that the eigenvalues of this matrix are not all positive, hence it is not CP, but it is still TP since the first row of $R$ corresponds to (100...000) which is the necessary condition.
\section{Details of the syndrome vector in Steane error correction}\label{section: syndromedetail}
As shown in the main text, we have four blocks instead of two blocks as in standard Steane correction. This also gives us extra information of the error source. For the sake of simplicity, we explain for $\gamma$ errors, and similar idea follows for $\bar{\gamma}$ errors.

Suppose I start in the initial states as specified in the main text, and before we implement the circuit, we were to measure the syndromes of each of the blocks in the occupation basis. For instance, if I had an error $\gamma_i$ for $i=\{1,\ldots n\}$ on the data block, I will get a syndrome bit $1$ for the data block on some $j$ bit $j=\{1,2,\ldots n-k\}$. In practice, I do not actually measure these syndromes, but for the sake of this explanation, we deem them as ``input syndrome vectors". In the table \ref{tab:inputs}, $1$ denotes a that one of the syndrome bit is flipped, while $0$ denotes that no syndrome bit is  flipped. We assume that a syndrome bit flip here means that there was an error in that block.  

\begin{table}[]
    \centering
    \vspace{10pt}
    \begin{tabular}{|c|c|c|c|}
    \hline
     $s_d$   & $s_2$ & $s_3$ & $s_4$  \\
     \hline
     1 & 0 & 0 & 0\\
     0 & 1 & 0 & 0\\
     0 & 0 & 1 & 0\\
     0 & 0 & 0 & 1\\
     \hline
    \end{tabular}
    \caption{Input table}
    \label{tab:inputs}
\end{table}
Table \ref{tab:outputs} shows the syndrome outputs on each block. Note that we do not actually measure $s_d$ since the data block holds the information, rather the goal is to use information from $s_2,s_3,s_4$ to find the syndrome on data block. 
\begin{table}[]
    \centering
    \vspace{10pt}
    \begin{tabular}{|c|c|c|c|}
    \hline
     $s_d$   & $s_2$ & $s_3$ & $s_4$  \\
     \hline
     1 & 0 & 1 & 1\\
     0 & 1 & 1 & 1\\
     1 & 1 & 0 & 1\\
     1 & 1 & 1 & 1\\
     \hline
    \end{tabular}
    \caption{Output table}
    \label{tab:outputs}
\end{table}
From table \ref{tab:outputs}, it is now obvious that if all three ancilla syndrome vectors are flipped, then there is no error on data, but if two of them are flipped, then there is an error on data block. This is captured by the Boolean function given in Eq.\ref{eq:boolean}. Moreover, since all four output strings are different even if we ignore the data syndrome, this means that we can distinguish between all four sources of syndromes, which is different from the standard qubit approach.

\bibliography{ref.bib}

\end{document}